\documentclass[sn-mathphys-num]{sn-jnl}

\usepackage{graphicx}%
\usepackage{multirow}%
\usepackage{amsmath,amssymb,amsfonts}%
\usepackage{amsthm}%
\usepackage{mathrsfs}%
\usepackage[title]{appendix}%
\usepackage[dvipsnames]{xcolor}%
\usepackage{textcomp}%
\usepackage{manyfoot}%
\usepackage{booktabs}%
\usepackage{algorithm}%
\usepackage{algorithmicx}%
\usepackage{algpseudocode}%
\usepackage{listings}%

\usepackage{graphicx}
\usepackage{csquotes}

\theoremstyle{thmstyleone}%
\newtheorem{theorem}{Theorem}
\newtheorem{lemma}[theorem]{Lemma}

\theoremstyle{thmstyletwo}%
\newtheorem{remark}{Remark}%

\theoremstyle{thmstylethree}%
\newtheorem{definition}{Definition}%

\begin{document}

\title{Spacetime games subsume causal contextuality scenarios}

\author*[1]{\fnm{Ghislain} \sur{Fourny}}\email{ghislain.fourny@inf.ethz.ch}

\affil*[1]{\orgdiv{Department of Computer Science}, \orgname{ETH Zurich}, \orgaddress{\street{Stampfenbachstrasse 114}, \city{Zürich}, \postcode{8092}, \country{Switzerland}}}

\abstract{We show that a category of causal contextuality scenarios with no cycles, unique causal bridges, and causally secured covers is equivalent to a category containing a subclass of the formerly published spacetime games, which generalize game theory to decisions arbitrarily located in Minkowski spacetime. This insight leads to certain constructs and proofs being shorter, simpler, and more intuitive when expressed in the spacetime game framework than in the causal contextuality scenario framework. The equivalence of categories and the modular structure of causal contextuality theory also implies that it is possible to build (pure) strategy sheaves, mixed strategy presheaves (equivalent to distribution presheaves) and empirical models on top of spacetime games: the obstruction to a global section in the presence of contextuality corresponds to the non-existence of a mixed strategy in the sense of the Nash game theory framework. This shows that the insights of both frameworks taken together can contribute positively to advancing the field of quantum foundations. }


\maketitle

\maketitle

\section{Introduction}

Quantum theory has been challenging the physics community for a century now. The question remains open of whether there exists an extension theory, as envisioned by Einstein, Podolsky, and Rosen (EPR) \cite{Einstein1935}, which is able not only to predict the same statistical distributions as the Born rule for repeated experiments---they are observed in Nature with very high precision---but also to predict each \emph{individual} measurement outcomes.

But the way such a theory can look like is constrained. Indeed, two assumptions \citep{Bell1964} are made for Bell inequalities to hold\footnote{Determinism is derived from these two with a short argument made by Bell \cite{Bell1964}, thus the companion adjective ``local-realist.''}:

\begin{itemize}
\item Freedom of choice: the choice of a measurement setting is independent of anything that could not have been caused by it. This implies \cite{Sengupta2025} (i) measurement independence: the choice of a measurement setting, by the experimenter, is exogeneous and, as such, independent of a hypothetical hidden variable, and (ii) parameter independence: a measurement outcome does not depend on a remote choice of measurement setting.
\item Local determinism\footnote{Bell in fact derives this assumption from parameter independence and of the application of the Born rule, which gives a probability of one when both settings are identical.}: a measurement outcome is completely determined by what happened in its past light cone.
\end{itemize}

Since there is experimental evidence \citep{Aspect1982} that Bell inequalities are broken by Nature, it follows that one of these two assumptions does not hold. Most of the quantum physics literature so far has focused on assuming freedom of choice and dropping local determinism. Little, but not nothing, has been published on the second possibility, namely, that freedom of choice does not hold.

Perhaps this is related to its name: ``free choice'' or ``free will'' \cite{Renner2011}, which can lead to strong or emotional philosophical opinions on the matter: we are scientists, but we are also human beings. Yet, it is so that freedom of choice is an assumption that can be formalized mathematically in terms of the statistical independence and/or the counterfactual independence (the two are distinct but related in subtle ways \cite{Fourny2020mp}) of random variables, and the scientific approach dictates that we must consider the possibility that freedom of choice does not hold with as much reason as the possibility that it does. The risk of failing to do so is that an extension theory exists but remains undiscovered because we do not commit any resources to finding it.

One of the most interesting developments in the quantum foundations is the realization that many surprising and disturbing aspects of quantum theory can be abstracted away from physical experiments, and defined purely in terms of which experiments a physicist decides to carry out, and which outcomes they observe. Furthermore, one does not need infinite lists of choices: EPR setups are already interesting with two agents, two settings per agent, and two outcomes per agent and setting (2-2-2). In 2019, we demonstrated \citep{Fourny2019} that one can model and view an experimental setup as a game between observers and the universe, Nature being an economic agent that minimizes action (principle of least action).

Later, Baczyk and Fourny \cite{Baczyk2023} gave a qualitative argument that the Nash equilibrium \cite{Nash1951} paradigm, which has been the mainstream framework of game theory for 75 years, corresponds to the conditions that lead to Bell inequalities:

\begin{itemize}
\item locality corresponds to imperfect information (dashed lines);
\item realism corresponds to the strategic form that can be built from a game in extensive form;
\item freedom of choice corresponds to unilateral deviations;
\item a local hidden-variable model corresponds to a mixed strategy (which is a statistical mixture of pure strategies) in the strategic form.
\end{itemize}

With this insight, one learns that any extension theory of quantum physics, with this game interpretation, cannot follow the Nash paradigm, and must follow instead a non-Nashian game theory paradigm, such as the Perfect Prediction Equilibrium or the Perfectly Transparent Equilibrium \cite{Fourny2020mp}.

Not being aware of our previous work \citep{Fourny2019}\cite{Fourny2020} of 2019 and 2020 as an antecedent, Abramsky et al. \cite{Abramsky2024} independently rediscovered in 2024 the idea that quantum experimental setups spread in spacetime can be seen as games between observers and the universe. They adapt a previous paper of 2011 \cite{Abramsky2011}, which we consider a true masterpiece of mathematics, on sheaf-theoretic structures from the flat case to the causal case. The central, original contribution of \cite{Abramsky2024} is that empirical models can be generalized to the causal case in such a way that non-local or more generally contextual quantum experiments exhibit obstructions to a global section, that is, also when measurements causally dependent on each other. This is done by composing a distribution functor with a strategy presheaf instead of the event sheaf. This provides a considerably stronger and timely \emph{quantitative} counterpart to our qualitative argument \cite{Baczyk2023}.

In this paper, we show that our framework is expressive enough to capture the game structures underlying \cite{Abramsky2024}, called causal contextuality scenarios, under some assumptions\footnote{In future work, we plan to generalize the equivalence of categories by relaxing some of these assumptions. In this paper we decided to focus on this already generic case.}: no cycles, unique causal bridges, and causally secured covers. Formally, under these assumptions, causal contextuality scenarios are equivalent, as a category, to a subcategory of spacetime games that we call alternative spacetime games. We provide the corresponding witness functor.

We conclude by showing that the bridge between the two formalisms is mutually beneficial and leads to a clearer understanding of a decision-theoretical view of quantum physics. In particular, we discuss how certain seemingly complex constructs, reasonings, or proofs involving causal contextuality scenarios become one-line arguments in the spacetime game framework. This includes a short argument that, under the assumptions made in this paper, the strategy presheaf is a sheaf since strategies can be constructed as sets of events using the strategic form \cite{Kuhn1950} similar to the flat case. Likewise, in the opposite direction, the major contribution of \cite{Abramsky2024} on empirical models, thanks to its modular functor-based structure, can be applied on top of alternating spacetime games, since they are structurally equivalent to the causal contextuality scenarios fulfilling our three assumptions.

\section{Background}

But before we do so, we are first going through the spacetime game framework with a few examples.

\subsection{Spacetime games and their extensive form}

Game theory has two main game structures \cite{Rubinstein1994}: the normal form and the extensive form. As explained by Baczyk and Fourny \cite{Fourny2020}\cite{Baczyk2023}, the normal form corresponds to spacelike-separated decisions, while the extensive form corresponds to timelike-separated decisions and even adaptive measurements. A general framework combining both, i.e., for decisions with arbitrary positions in spacetime, did not exist before 2019 to the best of our knowledge. If we are to consider generic quantum experimental setups with labs running local experiments placed at arbitrary positions in spacetime, and sending quantum states to one another---this corresponds to a fixed causal order in the process matrix framework \cite{Oreshkov2012}---, we need an extension and unification of the normal form and extensive form.

Thus, a new class of games called spacetime games with perfect information was introduced \cite{Fourny2019} \cite{Fourny2020}. It extends the semantics of game theory to special relativity. In spacetime games, the agents make decisions at various positions in Minkowski spacetime. Spacetime games can be seen as the least common denominator of games in normal form on the one hand (spacelike separation), and games in extensive form with perfect information on the other hand (timelike separation).

\begin{figure}
\begin{center}
\includegraphics[width=\textwidth]{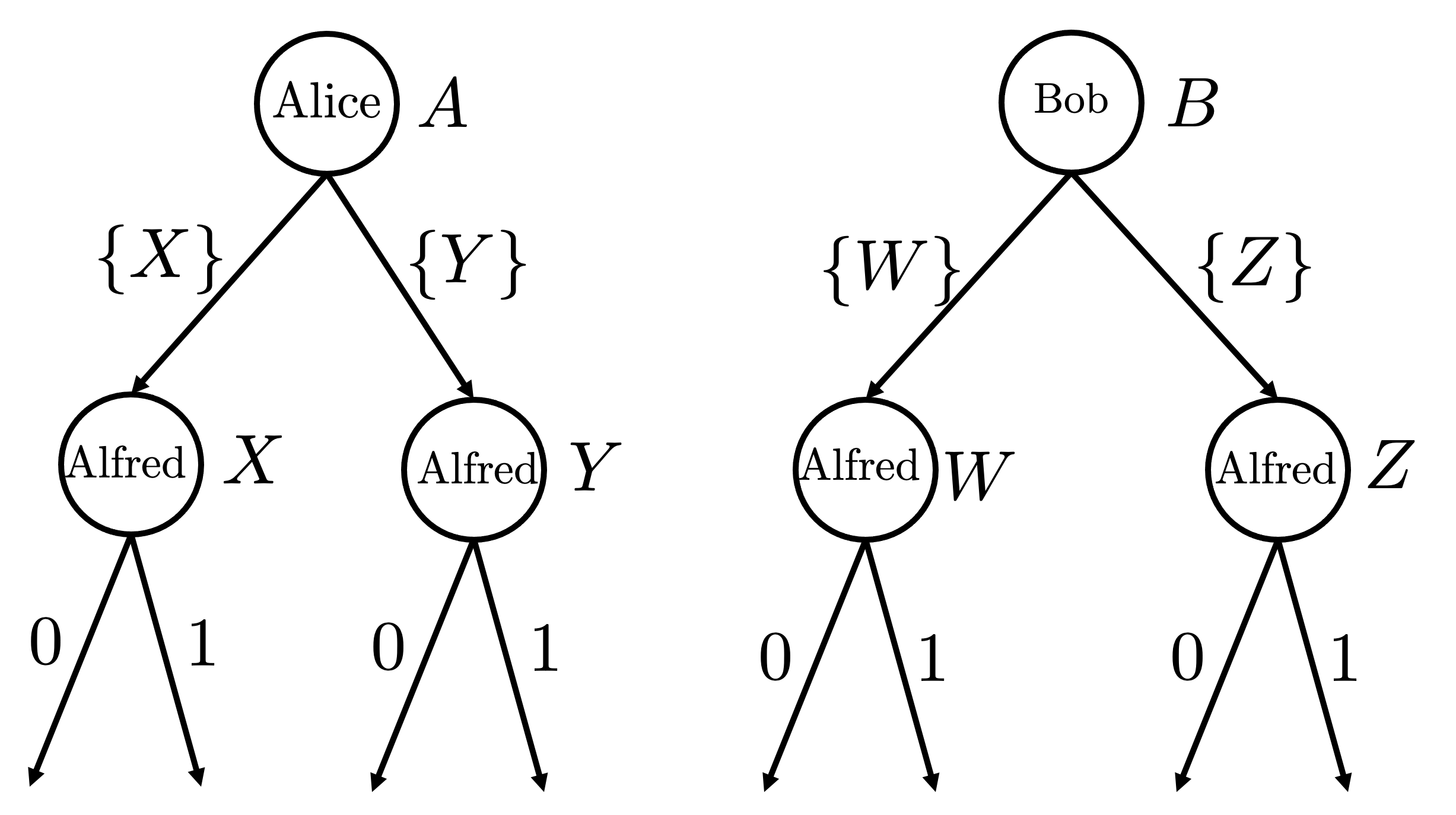}
\end{center}
\caption{A spacetime game with perfect information that represents a Bell experiment with two spacelike-separated observers. Alice and Bob each select a choice of measurement context (here, only singleton contexts). Alice can pick X or Y, and Bob can pick W or Z. Then, for each one of the two measurements selected, Alfred selects an outcome (0 or 1). Perfect information in spacetime games is suitable for the study of non-locality, a special case of contextuality.}
\label{fig1}
\end{figure}

The same documents explain how these games can be used to model quantum experiments as a game played between observers and nature. An example of a game corresponding to the Bell experiment is shown in Figure \ref{fig1}, where Alice and Bob are the observers, and Alfred, whom we like to visualize as a cat, represents nature.

A canonical injection of the class of spacetime games with perfect information into the class of games in extensive form with imperfect information is also given. Figure \ref{fig2} shows the game in extensive form with imperfect information that corresponds to the game of Figure \ref{fig1}. The two have equivalent semantics.

This is important because it means that spacetime games, when converted to games in extensive form with imperfect information, inherit several decades' worth of research, constructs, and proofs carried out on games in extensive form with imperfect information for the past 75 years. This is key to our claim that several constructs and proofs behind causal contextuality scenarios can be significantly simplified.  In other words, spacetime games provide an intuitive visual representation for quantum experimental setups but also a conversion algorithm to represent such setups as games in extensive form with imperfect information.

\subsection{Imperfect information and the dashed lines}

Imperfect information is visualized with dashed lines, which group nodes into equivalence classes called \emph{information sets}. The meaning of these dashed lines is that the agent who makes the decision at these nodes \emph{does not know} at which one of these nodes they are making their decision: they are not fully informed about the decisions made at ancestor nodes\footnote{These ancestor nodes need not all be in the past, they can also be spacelike-separated in the spacetime game. Since information cannot propagate faster than light, it is not possible to be informed about decisions made at spacelike-separated locations.} in the extensive form. This is why it is called imperfect information. Concretely, in Figure \ref{fig2}, Bob is not informed of Alice's decision ($\{X\}$ or $\{Y\}$), and this is why his two nodes are connected with a dashed line. This has to do with relativity theory: although Alice may have already made her decision, the information has not reached Bob yet because his decision is spacelike-separated from Alice's, as can be seen in the spacelike form of Figure \ref{fig1}, which is a causal dependency graph in the sense of relativity theory. Thus, this corresponds to a form of no-signaling.

The information sets of the extensive form in Figure \ref{fig2} correspond exactly to those of the spacetime game in Figure \ref{fig1} (which are singletons, one for each node because the spacetime game has perfect information). This is always the case for any spacetime games and their extensive form. The generalization of this canonical injection to all spacetime games, even spacetime games with imperfect information, is straightforward and explained in this paper.

\subsection{Imperfect information, locality, and non-contextuality}

The semantics of imperfect information and dashed lines have a subtle difference depending on whether one looks at the spacetime game or its extensive form. In a spacetime game with imperfect information, dashed lines physically represent \emph{non-contextuality}, connecting identical decisions made in different contexts, whereas in a game in extensive form with imperfect information (at least those built from a spacetime game), dashed lines represent either \emph{locality} or \emph{non-contextuality}. The spacetime game framework makes explicit that non-locality is a special case of contextuality: in non-local scenarios, the spacetime game has perfect information and the extensive form has imperfect information, while in more general contextuality scenarios, both of them have imperfect information.

\begin{figure}
\begin{center}
\includegraphics[width=\textwidth]{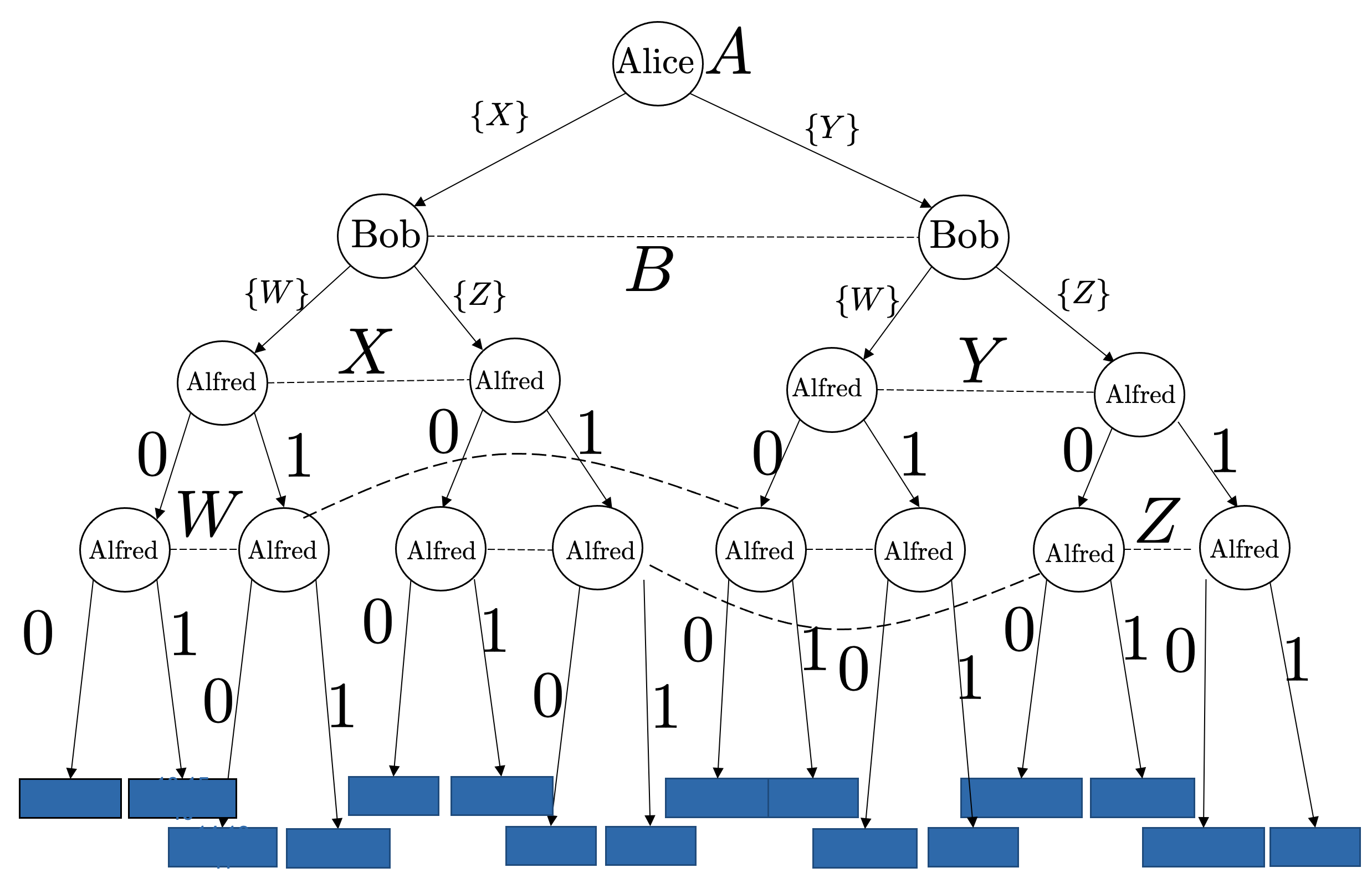}
\end{center}
\caption{The game in extensive form with imperfect information that corresponds to the spacetime shown in Figure \ref{fig1}. Alice and Bob have each one information set (A and B) for their choice of measurement setting. Alfred has four information sets X, Y, W, and Z, one for each measurement. For scenarios of interest to the quantum foundations, i.e., exhibiting non-locality or contextuality, the game in extensive form will typically have imperfect information, regardless of whether the spacetime game has perfect information (for the study of non-locality) or imperfect information (for the more general study of contextuality).}
\label{fig2}
\end{figure}

Because of the canonical injection of spacetime games into the class of games in extensive form with imperfect information, many well-studied concepts in the game theory literature apply to spacetime games. This includes flattening the strategic space with the (reduced) strategic form of a spacetime game \cite{Fourny2020}\citep{Kuhn1950}, computing Nash equilibria \cite{Nash1951}, mixed strategies (which we will see are also global sections of empirical models), etc.

\subsection{Causal contextuality scenarios}

Subsequent work by Abramsky et al. \cite{Abramsky2024} suggests modeling a subclass of such games, which they call causal contextuality scenarios, in terms of strategy presheaves. Thanks to a carefully designed modular construction based on a composition of the strategy presheaf with functors, a distribution presheaf is built, with (unsigned or signed) probabilistic or possibilistic semantics, depending on the selection of an underlying semi-ring. Then, contextuality can be witnessed by the obstruction to the existence of a global section in the distribution presheaf, given a set of compatible local distributions with respect to some cover.

The work \cite{Abramsky2024} is itself an adaptation of earlier work on the flat case \cite{Abramsky2011}, which brings causality into the picture.

As we show in this paper, certain causal contextuality scenarios---meaning, the underlying structure giving birth to the strategy presheaf---are equivalent to a subclass of spacetime games, which we call alternating spacetime games. This means that spacetime games recover causal contextuality scenarios as a special case. By transitivity (based on the claims made in \cite{Abramsky2024}), this means that by applying the functor connecting the associated strategy presheaf to the distribution presheaf, they can be used to recover the framework of Gogioso and Pinzani as a special case \citep{Gogioso2021}, as well as Hardy models \citep{Hardy1993}, PR boxes \citep{Popescu1994}, and GHZ models \citep{Greenberger1989}.

We also show that a flattened definition of the strategy presheaf based on strategic forms \citep{Kuhn1950} emerges naturally and that the sheaf property holds\footnote{for the strategy presheaf, not for the distribution presheaf}, and that temporal contextuality scenarios are subsumed by games in extensive form with perfect information.

\subsection{Spacetime games}

\subsubsection{Definition}

We now give a formal definition of spacetime games. It is equivalent to that previously published \citep{Fourny2019} \citep{Fourny2020}, but in a more compact form (see also \citep{Baczyk2023}). In this definition, we do not require perfect information. Imperfect information has the meaning proposed by Kuhn \cite{Kuhn1950} and extends naturally to spacetime games without affecting the injection into the class of games in extensive form with imperfect information\footnote{This is because merging nodes into information sets in the spacetime game translates into merging the corresponding information sets in the extensive form.} Spacetime games with perfect information allow for the study of non-locality, while spacetime games with imperfect information allow for the more general study of contextuality.

\begin{definition}[Spacetime game]
A spacetime game $\mathcal{G}$ is a tuple $( \mathcal{N}, \mathcal{R},\mathcal{P},\mathcal{A}, \rho, \chi, $ $\sigma, \mathcal{I}, \mathcal{Z}, u)$ where

\begin{itemize}
\item $(\mathcal{N}, \mathcal{R})$ is a directed acyclic graph of decision nodes, $\mathcal{N}$ being the  nodes and $\mathcal{R}$ being the edges.
\item $\mathcal{P}$ is a set of players and $\mathcal{A}$ is a set of actions (in our context, actions are measurement settings and measurement outcomes)
\item Nodes are labeled with players: $\rho\in\mathcal{P}^\mathcal{N}$
\item Nodes are labeled with sets of actions: $\chi\in\mathcal{P}(\mathcal{A})^{\mathcal{N}}$
\item Edges are labeled with actions: $\sigma\in\mathcal{A}^\mathcal{R}$ labels each edge $(N,M)$ with an action $\sigma(N,M)\in\chi(N)$.
\item A partition of $\mathcal{N}$ into information sets: a set $\mathcal{I}$ of subsets of $\mathcal{N}$ that form a partition, in a way that is compatible with $\rho$ and $\chi$, meaning that $$\forall i\in\mathcal{I},\forall N,M\in i, \rho(N)=\rho(M) \wedge \chi(N)=\chi(M)$$
\item $\mathcal{Z}\in(\mathcal{A}\cup\{\bot\})^\mathcal{I}$ is a set of outcomes, also called complete histories, which are partial functions from information sets to actions. They must obey some constraints implied by the game structure. These constraints are defined further down in Section \ref{section-histories}. The domain of definition, also called path, of an outcome is the set of information sets not mapped to $\bot$, which is the symbol used for unassigned actions.
\item $u=(u_i)_{i\in P}$ with, for each $i$, $u_i\in\mathbb{R}^\mathcal{Z}$ is a family of utility functions mapping outcomes to payoffs, with one payoff for each outcome and player.
\end{itemize}
\end{definition}

The payoffs associated with the outcomes are part of a game's semantics but are not the focus of the present paper, which focuses on the game structure. They are also not part of any causal contextuality scenario and are ignored for the purpose of the equivalence of categories\footnote{In the sense that games with the same structure and different payoffs are isomorphic.}.

\subsubsection{Further examples}

Figure \ref{fig3} shows an example of a spacetime game with imperfect information (the dashed lines have the same meaning as in Figure \ref{fig2}) and Figure \ref{fig4} shows the corresponding extensive form with imperfect information.

\begin{figure}
\begin{center}
\includegraphics[width=\textwidth]{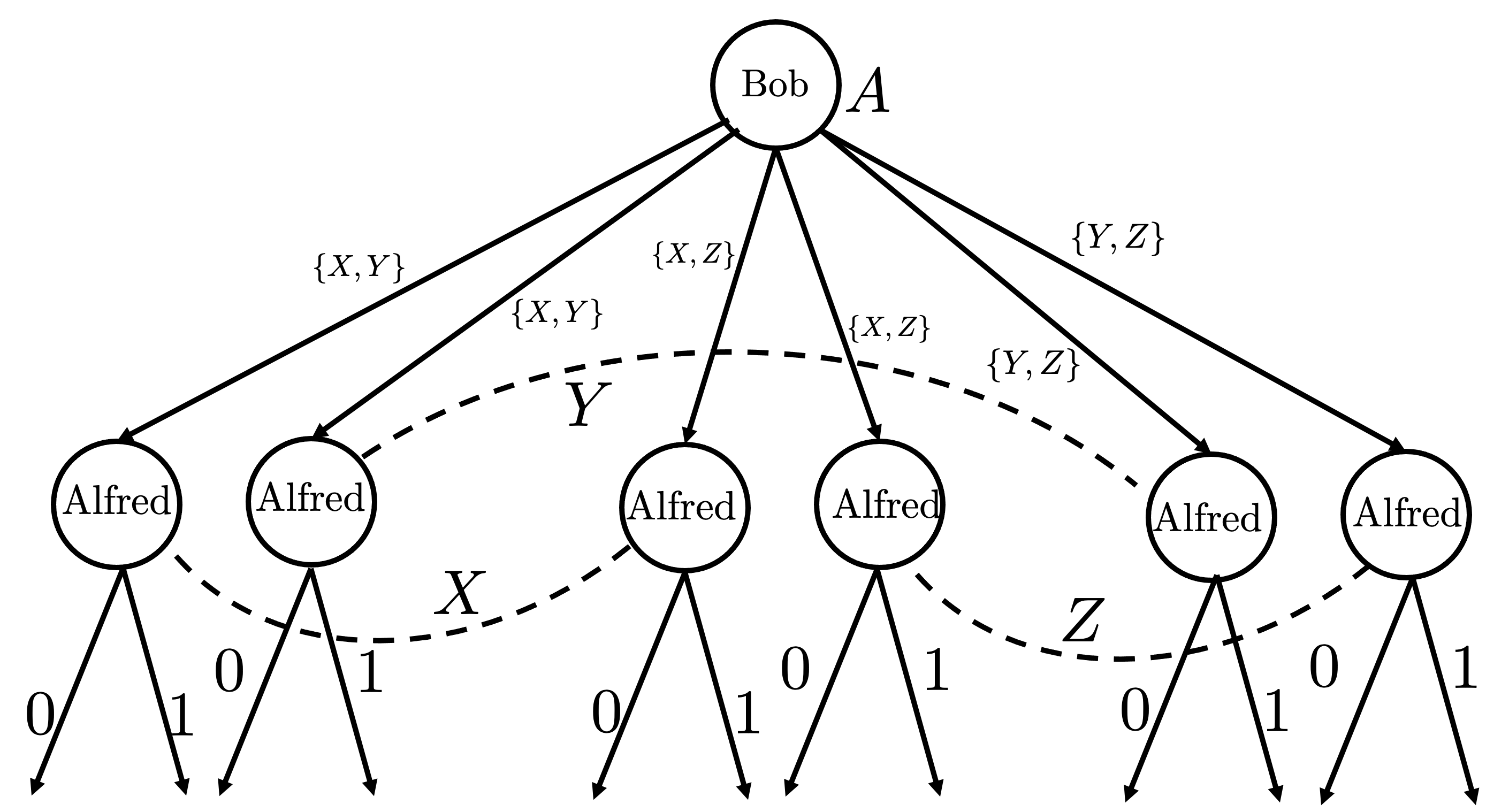}
\end{center}
\caption{A spacetime game with imperfect information that represents a non-contextuality experiment with two settings to pick out of three. Bob has one information set with three actions (the three contexts XY, XZ, and YZ). Alfred has 3 information sets X, Y, and Z for each measurement, with two actions each, 0 and 1. This is also known in the literature as a cyclic system of rank 3. Imperfect information in spacetime games captures the more general case of non-contextuality than the special case of locality.}
\label{fig3}
\end{figure}

Note that, to ease visualization, actions on the leaf nodes are shown as ``virtual'' edges with no destination node and labeled with the possible actions at these leaf nodes, but these edges are purely visual and do not belong to $\mathcal{R}$; instead, the edges leaving a leaf node $N$ correspond to the elements of $\chi(N)$.

\begin{figure}
\begin{center}
\includegraphics[width=\textwidth]{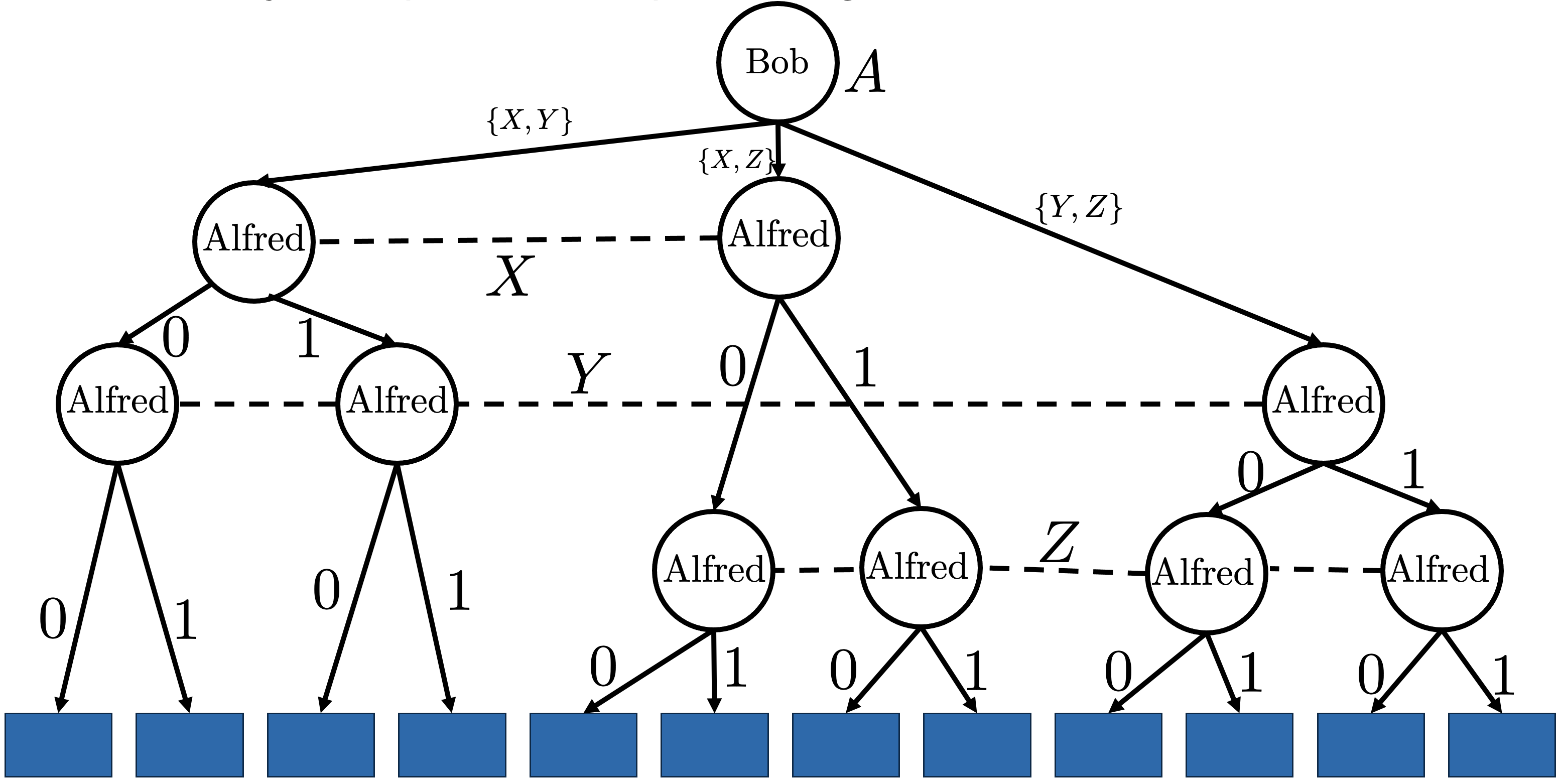}
\end{center}
\caption{The game in extensive form with imperfect information that corresponds to the spacetime game shown in Figure \ref{fig3}. It has the same information sets: XY, YZ, and XZ for Bob, and X, Y, and Z for Alfred.}
\label{fig4}
\end{figure}

\clearpage
Figure \ref{fig5} shows an example of a spacetime game with perfect information that has a causal bridge (adaptive measurement) and the corresponding extensive form is shown in Figure \ref{fig6}. Spacetime games can be much more complex and mix spacelike-separation as in Figure \ref{fig1}, timelike-separation as in Figure \ref{fig5}, and imperfect information for explicit contexts as shown in Figure \ref{fig3}.

\begin{figure}
\begin{center}
\includegraphics[width=\textwidth]{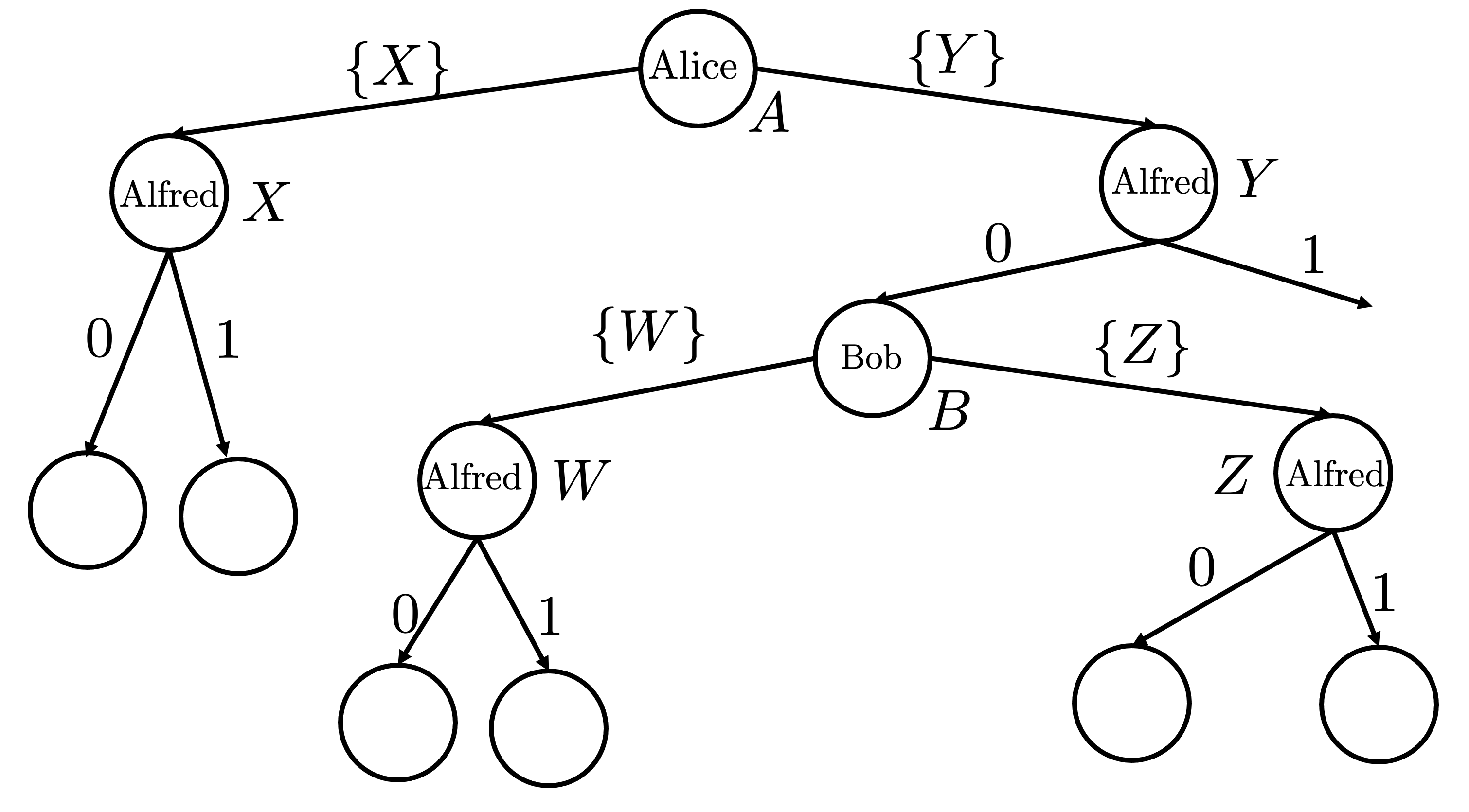}
\end{center}
\caption{A spacetime game with imperfect information that shows an example of causal bridge: Bob only carries out his experiment if Alice measured Y and obtained an outcome of 0.}
\label{fig5}
\end{figure}

\begin{figure}
\begin{center}
\includegraphics[width=\textwidth]{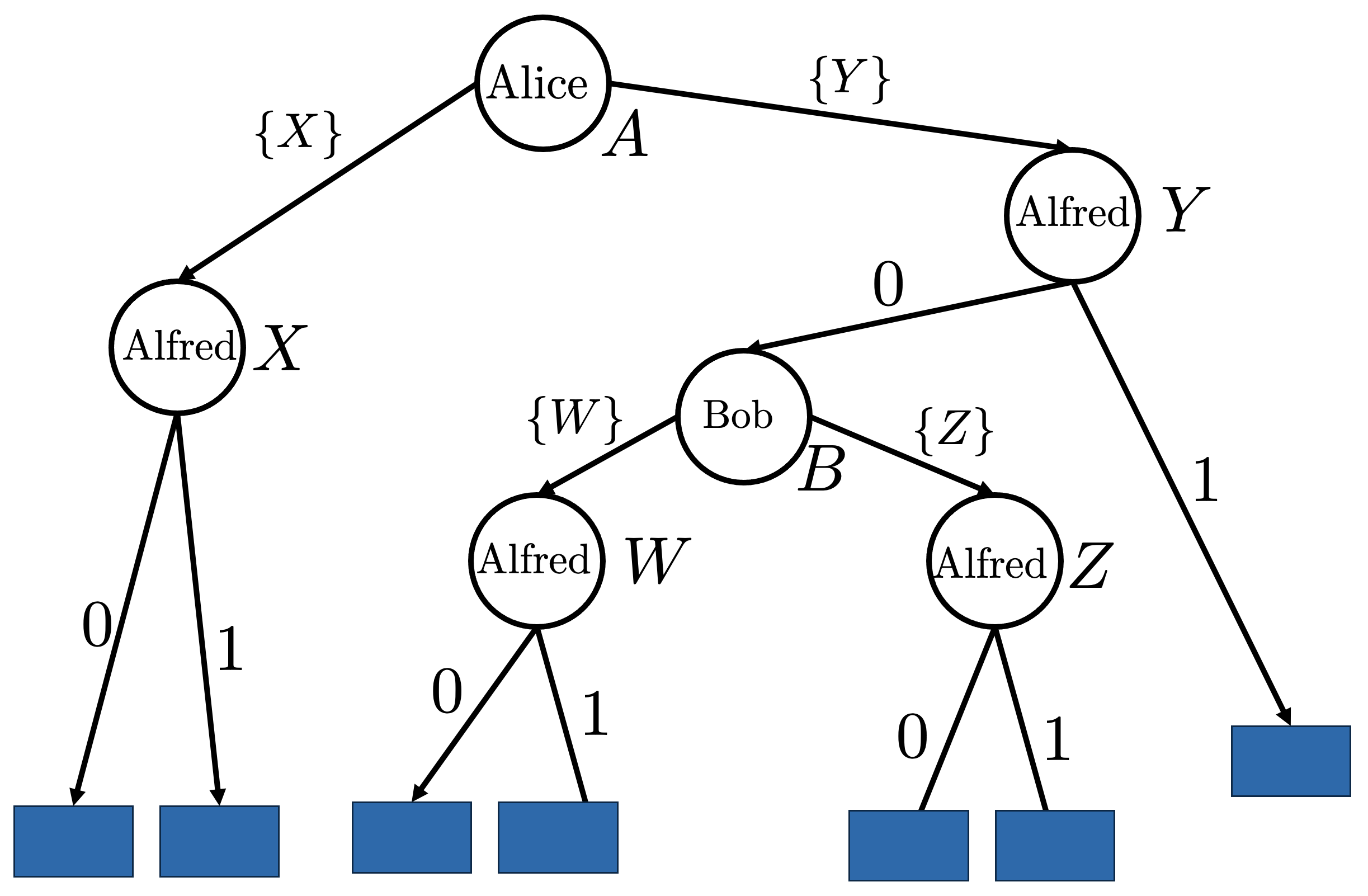}
\end{center}
\caption{The game in extensive form with, in this case, perfect information, that corresponds to the spacetime game shown in Figure \ref{fig5}. It happens to have the same tree shape because there is no spacelike separation. Such games correspond to temporal contextuality scenarios.}
\label{fig6}
\end{figure}

\clearpage
Since non-locality is a special case of contextuality, the Bell experiment can also be represented with a single player picking a context with two measurements, in a cyclic manner (XW, WY, YZ, ZX). The corresponding game is shown in Figure \ref{fig7}. The corresponding game in extensive form is shown in Figure \ref{fig8} and it is almost the same as the game shown in Figure \ref{fig2}, with the observer players merged into one, which amounts to looking at the Cartesian product of the Hilbert spaces.

\begin{figure}
\begin{center}
\includegraphics[width=\textwidth]{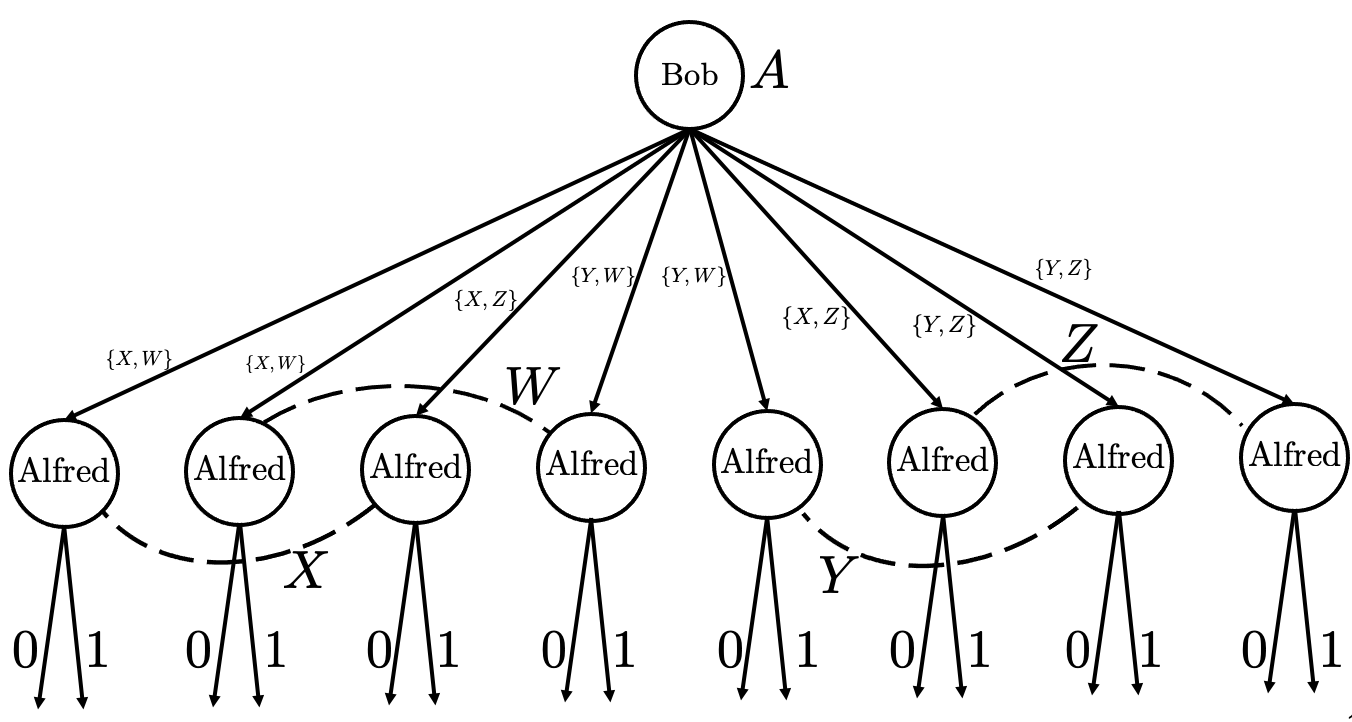}
\end{center}
\caption{A spacetime game with imperfect information that represents a Bell experiment as a non-contextuality experiment instead of a locality experiment, with a single observer player picking a context (A) of two measurement settings out of four possible contexts, and four information sets (X, Y, W, Z) played by nature. This is also known in the literature as a cyclic system of rank 4.}
\label{fig7}
\end{figure}

\begin{figure}
\begin{center}
\includegraphics[width=\textwidth]{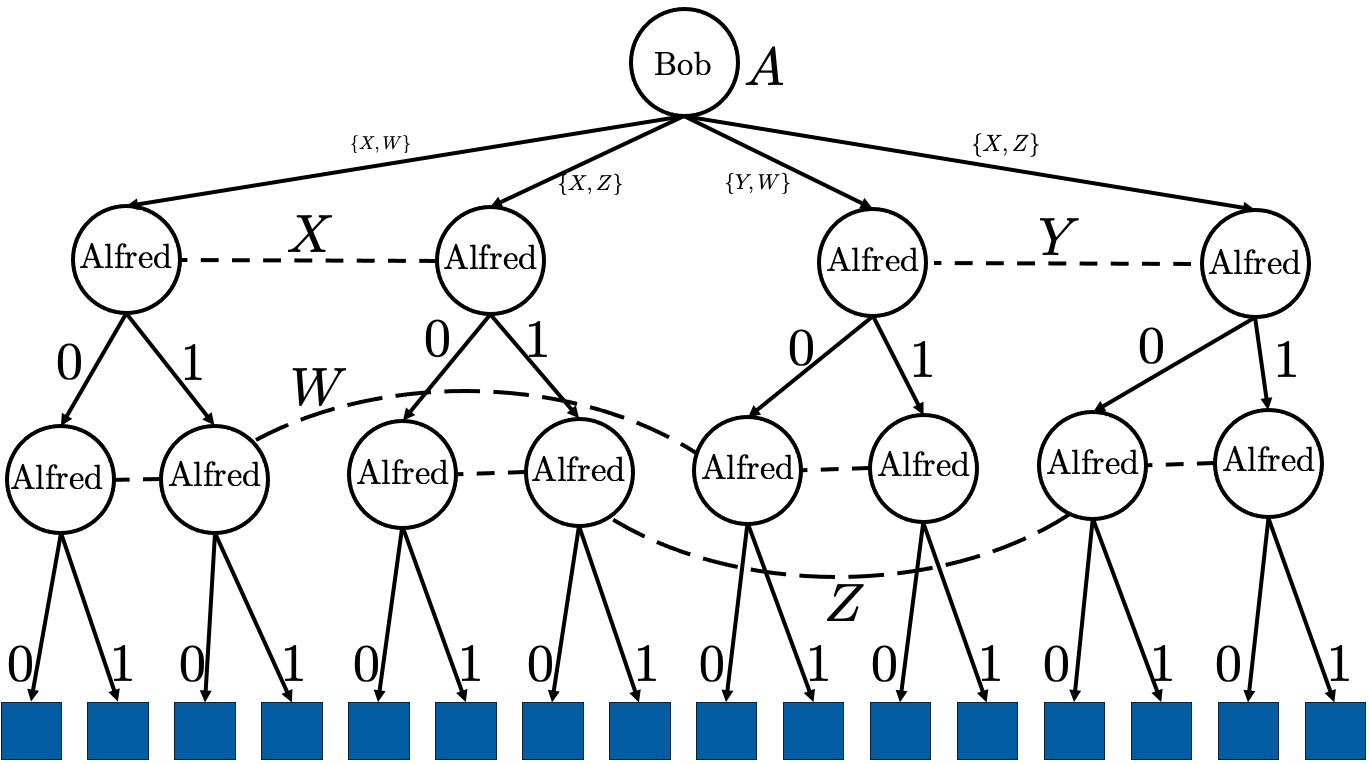}
\end{center}
\caption{The game in extensive form with imperfect information that corresponds to the spacetime game shown in Figure \ref{fig7}. It has the same information sets: A for Bob, and X, Y, W, and Z for Alfred.}
\label{fig8}
\end{figure}
\clearpage

Figure \ref{fig11} shows a game that corresponds to the GP example given in \cite{Abramsky2024}. The first observer Alice chooses a measurement X or Y and obtains a result. Then, the second observer chooses a measurement Z or W and obtains a result. On the spacetime form of the game (which in this case is also the extensive form because it is a temporal contextuality scenario), it looks as if the outcome of Bob's measurement (chosen by Alfred) depends not only on Alice's choice of measurement but also on the outcome obtained by Alice. However, if one computes the \emph{reduced} strategic form of the game \cite{Kuhn1950}, the inconsistent combinations (e.g., $X=0$ and $Z_{X,1}=0$), which are never activated, get eliminated so that Alfred's reduced strategy only needs to involve four choices, not eight. This is consistent with \cite{Abramsky2024}'s statement that Alfred's strategy need only each combination of Alice's and Bob's choices to 0 or 1.

\begin{figure}
\begin{center}
\includegraphics[width=\textwidth]{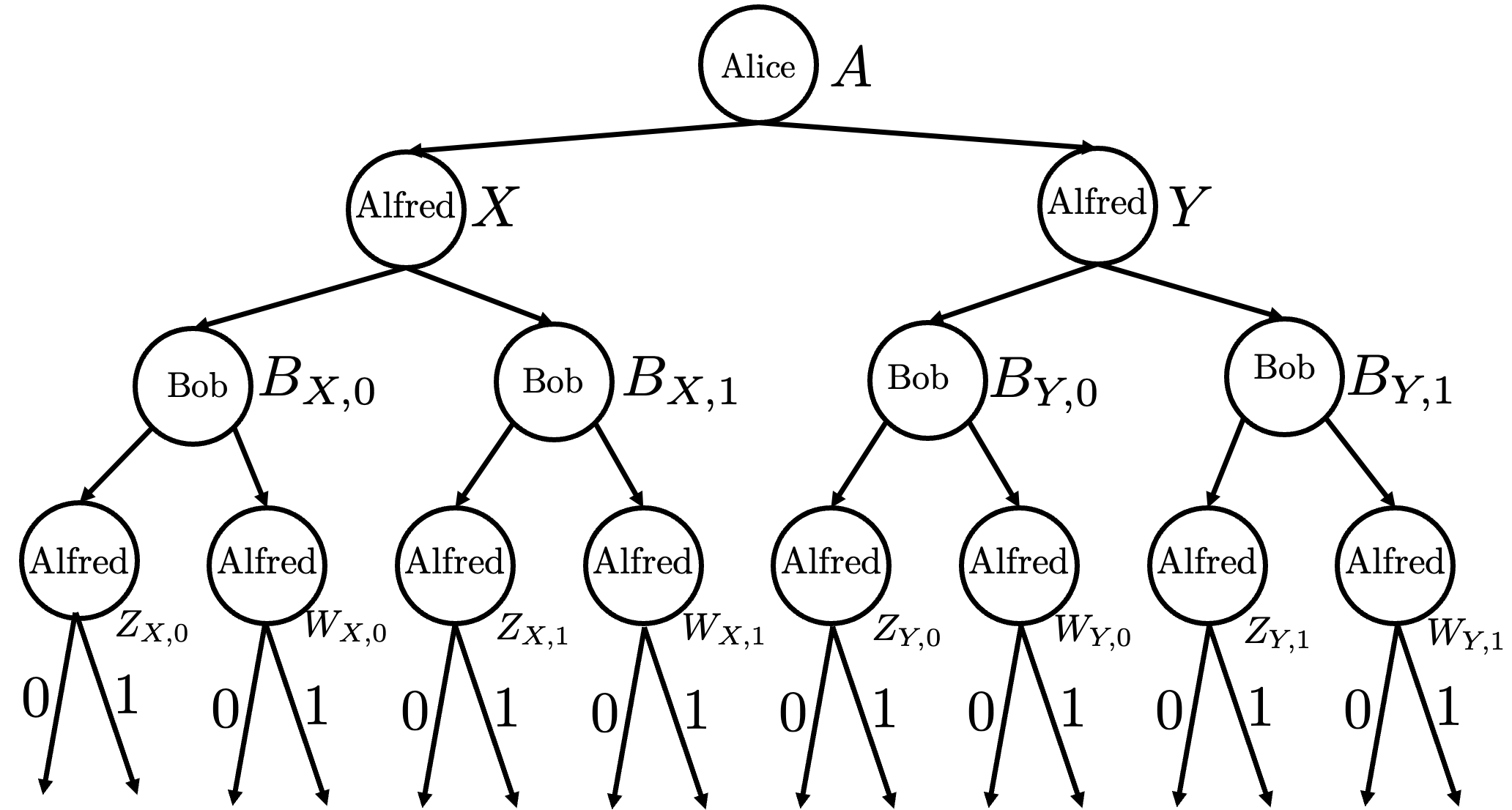}
\end{center}
\caption{A spacetime with perfect information corresponding to a GP scenario.}
\label{fig11}
\end{figure}

\clearpage

Figure \ref{fig12} shows a game\footnote{This scenario does not have unique causal bridges and is thus not included in this paper's categorical equivalence, however an equivalent spacetime game structure still exists as shown here. We plan to later show an extended categorial equivalence also for scenarios with multiple causal bridges.} corresponding to an other example given in in \cite{Abramsky2024}, which is the implementation of an OR gate on top of a GHZ scenario (Anders-Browne). The choice by Alice of the two inputs were flattened as a cyclic system of rank 4, but could also be expressed with two spacelike-separated agents as shown on Figure \ref{fig1} (which is equivalent to the cyclic system of rank 4 shown on Figure \ref{fig7}).

\begin{figure}
\begin{center}
\includegraphics[width=\textwidth]{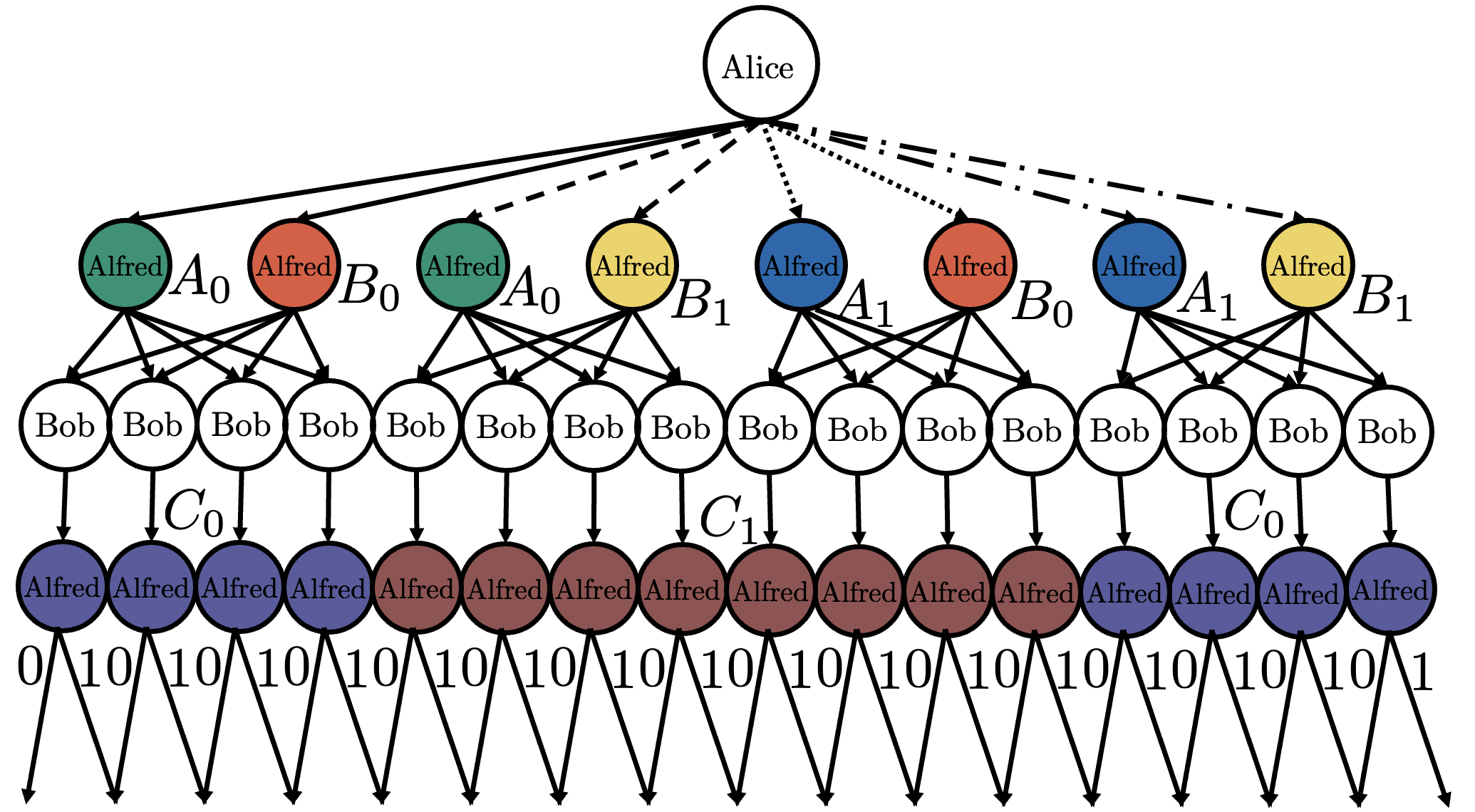}
\end{center}
\caption{A spacetime with perfect information corresponding to the implementation of an OR gate on top of a GHZ scenario. For ease of read, information sets are marked with colors instead of dashed lines. Also, also for the ease of read, the context labels were replaced with dashed lines. The labels can be deduced from the local simplicial complex that arises at each node played by Alice and Bob -- another nice visual feature of the spacetime game visualization.}
\label{fig12}
\end{figure}

\clearpage
\subsubsection{Notations}

In the remainder of this paper, whenever we refer to a game $\mathcal{G}$, we will also implicitly use the above notations ($\mathcal{N}, \mathcal{R}$, $\mathcal{P}$, $\mathcal{A}$, $\rho$, $\chi$, $\sigma$, $\mathcal{I}, u$) without explicitly defining them every time. This is because we only ever consider one game at a time, and it eases the read. When a second game is considered, we prime the letters: a game $\mathcal{G'}$, we will also implicitly use ($\mathcal{N'}, \mathcal{R'}$, $\mathcal{P'}$, $\mathcal{A'}$, $\rho'$, $\chi'$, $\sigma'$, $\mathcal{I}'$, u') without explicitly defining them every time. 

For any player $p\in\mathcal{P}$, we write $\mathcal{N}_p=\rho^{-1}(p)$. Thus, $\mathcal{N}$ can be organized in a family of disjoint partitions $(\mathcal{N}_p)_{p\in\mathcal{P}}$.

We write $N\smile M$ as a convenient shorthand for $(N,M)\in\mathcal{R}$.

We define for any node N $$\textsf{precedessors}(N)=\{M\in\mathcal{N}| M\smile N\}$$ $$\textsf{successors}(N)=\{M\in\mathcal{N}| N\smile M\}$$

We also define for any information set $i$ $$\textsf{precedessors}(i)=\bigcup_{N\in i}\textsf{predecessors}(N)$$ $$\textsf{successors}(i)=\bigcup_{N\in i} \textsf{successors}(N)$$.

We also use a slight abuse of notation whenever all nodes in the same information set share a common property: for example, $\rho(i)$ and $\chi(i)$ where $i$ is an information set would normally denote a singleton set, but we use them to denote the single element in that singleton set instead, for convenience.  Accordingly, for any player $p\in\mathcal{P}$, we write $\mathcal{I}_p=\{i\in\mathcal{I} | \rho(i)=p\}$. Thus, $\mathcal{I}$ can be organized in a family of disjoint partitions $(\mathcal{I}_p)_{p\in\mathcal{I}}$.

We also use, for convenience, a function $\iota\in\mathcal{I}^\mathcal{N}$ that maps each node to the information set it belongs to.

Conversely, we will write $z(N)$ for $z(\iota(N))$ where $z$ is an outcome (complete history) and $N$ is a node.

When we write the statement $\sigma(N,M)=a$ for some $a$, it implicitly implies $N\smile M$, which allows more concise formulas.

\subsubsection{Histories}
\label{section-histories}
We now give a compact definition of the histories of a spacetime game (which is a compact form of that given in \cite{Fourny2020}\cite{Baczyk2023}):

\begin{definition}[Histories]
Given a spacetime game in extensive form $\mathcal{G}$, a history is an assignment from $\mathcal{I}$ to $\mathcal{A}\cup\{\bot\}$ that only assigns an action to activated information sets. An information set is activated if it has a node such that the label of each one of its incoming edges is equal to the action assigned (by that same history) to the information set containing the source node of that edge:
$$\forall i\in\mathcal{I},(z(i)\neq\bot\implies \exists N\in i, \forall M\in \textsf{predecessors}(N), z(M)=\sigma(M,N))$$ 

The set of all histories is denoted $\mathcal{H}$.
\end{definition}

In \cite{Fourny2020}, we use the term ``contingency coordinates'' of a node $N$ to denote the minimal history not containing $N$ that coincides with any history in which $N$ is assigned an action (i.e., not $\bot$). In this paper, we use the term ``causal bridge'' \citep{Fourny2014}\footnote{We take the terminology from the work of Dupuy \citep{Dupuy1992}\citep{Dupuy2000}.} of a node $N$ to specifically designate the predecessors of $N$ and the labels on the edges from the predecessors to $N$. The two differ in that the contingency coordinates involve the entire ancestry, while the causal bridge only involves (directly connected, one level up) predecessors.

A history is complete if the implication becomes an equivalence. In a valid game, the set of outcomes is the set of complete histories.

\begin{definition}[Complete histories, and validity of a spacetime game]
\label{definition3}
A spacetime game in extensive form $\mathcal{G}$ is valid if the set of outcomes $\mathcal{Z}$ matches the set of complete histories. A complete history is defined as a history that assigns an action to all the activated information sets.
$$\forall i\in\mathcal{I}, (z(i)\neq\bot\iff \exists N\in i, \forall M\in \textsf{predecessors}(N), z(M)=\sigma(M,N))$$ 
\end{definition}

We always assume that spacetime games are valid, as we consider validity to be part of their definition. Given a complete history (outcome) $z$, we write:

$$\textsf{support}(z)=\{i\in \mathcal{I}, z(i)\neq\bot\}$$

Without loss of expressive power, we assume that games do not contain unused information sets, i.e., there is no information set that never gets activated; otherwise, unused information sets can be pruned from the game.\footnote{The same assumption will be made for causal contextuality scenarios: There is no measurement that never gets activated, otherwise such measurements can be pruned from the causal contextuality scenario.} Likewise, we assume there are no unused actions in $\mathcal{A}$.

\clearpage
\section{Results}

We will now build a category corresponding to a subclass of two-player spacetime games (which we call ``alternating'') as well as a category of causal contextuality scenarios with no cycles, unique causal bridges and causally secured covers, and show that these two categories are equivalent thanks to a witness functor.

Thus, spacetime games and causal contextuality scenarios are equivalent for a broad class of physically meaningful experimental protocols.

Complete histories under the two frameworks are also in bijection, and histories in causal contextuality scenarios correspond to those histories in spacetime games that have a domain with even cardinality (in other words, incomplete histories in causal contextuality scenarios always stop with nature giving an outcome for the chosen measurement, and never end abruptly with a measurement chosen but not carried out).

\subsection{Alternating spacetime games}

\subsubsection{Definition}

Spacetime games, but also games in normal or extensive form with or without perfect information, are very general and cover many different scenarios, from microeconomics to leisure (Chess, etc). Here, we are looking at specific kinds of spacetime games in which experimenters and the universe alternate in choosing settings and outcomes. We thus now give a characterization of a particular relevant subclass of spacetime games called ``alternating spacetime games'', for which we will show that there is a mapping with the corresponding causal contextuality scenarios.

Note that, since causal contextuality scenarios do not specify which observer performs which experiment, we call all our observers Bob for the mapping. If a causal contextuality scenario is extended with information on the identity of the observers (e.g. Alice, Bob) then they can be marked accordingly in the corresponding game.

In the definition of the alternating property, we interpret the game's assumed features in terms of quantum-experimental protocols, to facilitate the understanding of the mapping that will be given further down.

\begin{definition}[Alternating properties]
A spacetime game $\mathcal{G}$ is an alternating game if it is valid and it fulfills the following alternating properties:
\begin{enumerate}
\item[2-PLAYERS] It has two players (we call them Alfred, who is nature, and Bob, who is the observer).
\item[BIPARTITE] The graph $(\mathcal{N}, \mathcal{R})$ always connects a node played by Alfred to a node played by Bob, or a node played by Bob to a node not played by Alfred. In other words, the graph structure of the spacetime game is a bipartite graph if we ignore the direction of the edges.
\item[EVEN] All root nodes are played by Bob. All leaf nodes are played by Alfred.
\item[BOB-S] All information sets played by Bob are singletons. It means Bob is fully informed about decisions in his causal past, even in situations in which he can carry out the same experiment under different circumstances.
\item[BOB-A] At any of Bob's nodes, all available actions are used on at least one edge.\footnote{One can also consider allowing for at most one unused action in $\chi(N)$ when looking at causal contextuality scenarios with contexts in which the antichain property does not hold.}
$$\forall t\in\mathcal{N}_B, \forall C \in\chi(t),\exists n\in\mathcal{N}_A, \sigma(t,n)=C$$
\item[BA1] Each node played by Alfred has exactly one parent node:
$$\forall n\in\mathcal{N}_A, \exists! t\in\mathcal{N}_B, t\smile n$$
\item[BA2] Given a node $N$ played by Bob, two nodes (played by Alfred) connected to $N$ with the same label (it is a measurement context) must be in different information sets (these are measurement settings). In other words, distinct nodes for the same measurement in the same context would be superfluous.
$$\forall t\in\mathcal{N}_B, \forall n,m\in \textsf{successors}(t), \sigma(t,n)=\sigma(t,m)\implies \iota(n)\neq\iota(m)$$
\item[AB1] All nodes in the same information set played by Alfred\footnote{This trivially applies to Bob, too.} have the same outgoing edges: same labels, same destination nodes. In other words, the causal future of a measurement does not depend on the context in which it was carried out. 
$$\forall x\in\mathcal{I}_A, \forall n,m\in x, \textsf{successors}(n)=\textsf{successors}(m)$$
$$\forall x\in\mathcal{I}_A, \forall n,m\in x, \forall u\in\textsf{successors}(x), \sigma(n,u)=\sigma(m,u)$$
\item[AB2] Two distinct nodes played by Bob cannot have the same causal bridge.
$$\forall t, u\in \mathcal{N}_B,$$ $$(\textsf{predecessors}(t)=\textsf{predecessors}(m)\wedge \forall n\in\textsf{predecessors}(t), \sigma(n,t)=\sigma(n,u)) \implies t=u$$
\end{enumerate}
\end{definition}

\begin{remark}The height of an alternating spacetime game is always even, nodes at even depths (0, 2, ...) are played by Bob, and nodes at odd depths (1, 3, ...) are played by Alfred. 
\end{remark}

In the remainder of this paper, we standardize the two players of alternating games with Bob playing at the root and Alfred playing at the leaves, denote the partitions $\mathcal{N}_B,\mathcal{N}_A,\mathcal{I}_B,\mathcal{I}_A$, etc, and omit the set of players $\mathcal{P}$ from the definition. This is for simplicity and without loss of generality, as it would be straightforward to give the two players arbitrary names by including $\mathcal{P}$ explicitly.

We denote as $\text{parent}(n)$ the unique parent of $n$ for any $n\in\mathcal{N}_A$ (rule BA1).
For $x\in\mathcal{I}_A$ and $n\in\mathcal{N}_B$, we extend the $\smile$ notation to $x \smile n$ when $\forall m\in x, m\smile n$ as well as $\sigma(x,n)=\sigma(m,n)$ for any $m\in x$ and $\text{successors}(x)=\text{successors}(m)$ for any $m\in x$ (Rule AB1).

\subsubsection{Non-locality as a special case of contextuality}

AB2 amounts to saying that we assume that the measurements carried out by Bob under the same enabling circumstances are grouped into a unique experiment where Bob has multiple contexts to choose from. Note that this does not cause any loss of generality because this amounts to saying that non-locality is a special case of contextuality. See for example Figure \ref{fig7} (which satisfies AB2), which groups the local experiments of Figure \ref{fig1} (which does not satisfy AB2) into a single one.

The purpose of alternating spacetime games is to characterize a subset of spacetime games that is equivalent (in terms of category theory) to certain causal contextuality scenarios. When using the spacetime game framework for discussing non-locality, it is more desirable to consider the (factored) game in Figure \ref{fig1} rather than the (expanded) game in Figure \ref{fig7}, knowing that the factored game can be considered equivalent to the same causality contextuality scenario as the one associated with the expanded game by the mapping.

\subsubsection{Bob's actions as the facets of a simplicial complex}
\label{section-rename-bob-actions}
Without loss of generality, we can canonically rename Bob's actions in an alternating spacetime game to facilitate their interpretation as contexts, like so:

\begin{itemize}
\item Each edge from a node played by Bob to a node played by Alfred is labeled with a non-empty set (the context) of information sets (the settings) played by Alfred.
\item For any node $t$ played by Bob, there are as many outgoing edges from $t$ as there are context-setting pairs $(x,C)$ such that $x\in C$ and $C\in \chi(t)$. The edge corresponding to $(x,C)$ must go to a node $n\in x$ and must be labeled with $\sigma(t,n)=C$.
\end{itemize}

In other words, Bob's actions are contexts, and each context-action activates, via edges labeled with this context, the measurement settings in that (possibly empty) context.

It is always possible to rename Bob's actions such that this is the case: Given a node $t$ played by Bob and an action $\hat{C}\in\chi(t)$, we rename $\hat{C}$ to $C=\{\iota(n) |\sigma(t,n)=\hat{C}\}$. Because of BA2, the renaming is bijective.

Formally, it means that $\chi(t)$ can be seen as the set of the facets of an abstract simplicial complex, the vertices of which are the children information sets of $t$ and with the constraint $\forall n\in\textsf{successors}(t), \iota(n)\in\sigma(t,n)$.

\clearpage

\subsubsection{Category of alternating spacetime gamess}

We now define the category $\textbf{Game}$ of alternating spacetime games. For this purpose, we first define the morphisms. As the intent is to show it is equivalent to a category of causal contextuality scenarios, we define morphisms on the same lines as \cite{Abramsky2019} with the idea that the decisions of the universe player in a game can be simulated by another game into which the observer decisions are fed, in the same way a human can sit in front of two electronic chessboards and let the computers play against each other by alternatively feeding the output of one side on the other side.

\begin{definition}[Game morphism]
Consider two spacetime games $\mathcal{G}=( \mathcal{N}, \mathcal{R},\mathcal{A}, \rho, \chi, $ $\sigma, \mathcal{I}, \mathcal{Z}, u)$ and $\mathcal{G'}=( \mathcal{N'}, \mathcal{R'},\mathcal{A'}, \rho', \chi', $ $\sigma', \mathcal{I}', \mathcal{Z'}, u')$. A morphism $\gamma:\mathcal{G'}\rightarrow\mathcal{G}$ is a pair $(\nu',\beta)$ where:

\begin{itemize}
\item $\nu'$ is a function from $\mathcal{N}_{A}$ to $\mathcal{N'}_A$,
\item $\{\beta_x\}_{x\in\mathcal{I}_A}$ is a family of functions from $\chi'(\nu'(x))$ to $\chi(x)$.
\end{itemize}

With the following structural constraints:
\begin{enumerate}
\item $\nu'$ preserves information sets:
$$\forall x \in \mathcal{I}_A, \exists x' \in\mathcal{I'}_A, \forall n \in x, \nu'(n)\in x'$$
We write, with an abuse of notation, $x'=\nu'(x)$.

\item $\nu'$ preserves parents\footnote{The implication is in fact an equivalence because of the simplicial map rule}::
$$\forall n,m\in\mathcal{N}_A, \text{parent}(\nu'(n))=\text{parent}(\nu'(m))\implies\text{parent}(n)=\text{parent}(m)$$
This implies a function $\nu$ on $\mathcal{D}_{\nu}=\{\text{parent}(\nu'(n))|n\in\mathcal{N}_A\}\rightarrow\mathcal{N}_B$:
$$\nu: \text{parent}(\nu'(n))\mapsto\text{parent}(n)$$

\item Universe labels (measurement outcomes) are mapped consistently:
$$\forall n\in\mathcal{N}_A,\forall t\in'\mathcal{D}_{\nu}, \nu'(n)\smile t' \implies \beta_{\iota(n)}(\sigma'(\nu'(n),t'))=\sigma(n,\nu(t'))$$

\item For any $t' \in\mathcal{D}(\nu)$, $\nu'$ restricted to the children information sets of $\nu(t')$ forms a simplicial map from $\chi(\nu(t'))$ into $\chi'(t')$, and the individual nodes must be mapped consistently with this simplicial map:
$$\forall t'\in\mathcal{D}_{\nu},n\in\mathcal{N}_A, \nu(t')\smile n \implies \nu'(\sigma(\nu(t'),n))=\sigma'(t',\nu'(n))$$

\item All parents of observer nodes in the support of $\nu$ must be in the image of $\nu'$:
$$\forall x'\in \mathcal{I'}_A, \forall t'\in\mathcal{D}_{\nu}, x'\smile t' \implies \exists x\in\mathcal{I}_A, \nu'(x)=x'$$
\end{enumerate}
\end{definition}

There is an identity morphism that turns any game into itself:

\begin{definition}[Identity morphism for a game]
Consider one spacetime games $\mathcal{G}=( \mathcal{N}, \mathcal{R},\mathcal{P},\mathcal{A}, \rho, \chi, $ $\sigma, \mathcal{I}, \mathcal{Z}, u)$. The identity morphism $\text{Id}_\mathcal{G}:\mathcal{G}\rightarrow\mathcal{G}$ is the tuple $(\nu',\beta)$ where:

\begin{itemize}
\item $\nu'=\text{Id}_{\mathcal{N}_{A}}$.
\item $\beta_x=\text{Id}_{\chi(x)}$ for each $x\in\mathcal{I}_A$.
\end{itemize}
\end{definition}

\begin{remark}
It follows from the definition that:

\begin{itemize}
\item $\nu=\text{Id}_{\mathcal{N}_{B}}$.
\item The simplicial maps implied by $\nu'$ on the local contexts are the identity maps.
\end{itemize}
\end{remark}

\begin{lemma} The identity morphism for any game is a morphism.
\end{lemma}

\begin{proof}
All rules constraining a morphism become tautologies when replacing the functions with identity functions.
\end{proof}

\begin{lemma}[Left and right unit laws]
Consider two spacetime games $\mathcal{G}=( \mathcal{N}, \mathcal{R},\mathcal{A}, \rho, \chi, $ $\sigma, \mathcal{I}, \mathcal{Z}, u)$ and $\mathcal{G'}=( \mathcal{N'}, \mathcal{R'},\mathcal{A'}, \rho', \chi', $ $\sigma', \mathcal{I}', \mathcal{Z'}, u')$ as well as a morphism $\gamma:\mathcal{G'}\rightarrow\mathcal{G}$. Then $\gamma=\text{Id}_{G}\circ\gamma=\gamma\circ\text{Id}_{G'}$
\end{lemma}
\begin{proof} It follows from applying the left and right unit laws on the individual components.
\end{proof}

And morphisms can be composed like so:

\begin{definition}[Composition of morphisms on games]
Consider three spacetime games $\mathcal{G}$, $\mathcal{G'}$, and $\mathcal{G''}$ and two morphisms $\gamma_1:\mathcal{G'}\rightarrow\mathcal{G}$ with components $(\nu_1',\beta_1)$ (with implicitly $\nu_1$) and $\gamma_2:\mathcal{G''}\rightarrow\mathcal{G'}$ with components $(\nu_2'',\beta_2')$ (with implicitly $\nu_2'$). The morphism $\gamma_3=\gamma_1 \circ \gamma_2:\mathcal{G''}\rightarrow\mathcal{G}$ is defined as the tuple $(\nu_3'',\beta_3)$ (with implicitly $\nu_3$) where:

\begin{itemize}
\item $\nu_3''=\nu_2'' \circ \nu_1'$.
\item $\beta_{3,x}=\beta_{1,x} \circ \beta_{2,\nu_1'(x)}'$
\end{itemize}
\end{definition}

\begin{remark}
It follows that $\nu_3=\nu_1 \circ \nu_2'$, and the simplicial maps compose in the same way as $\nu3''$.
\end{remark}
\begin{lemma}
The result of morphism decomposition as defined above is a morphism.
\end{lemma}

\begin{proof}

Let there be $x\in\mathcal{I}_A$. $\gamma_1$ is a morphism, thus we have:
$$\forall n \in x, \nu_1'(n)\in \nu_1'(x)$$
$\gamma_2$ is a morphism, thus we have:
$$\forall n' \in \nu_1'(x), \nu_2''(n')\in \nu_2''(\nu_1'(x))$$
From which it follows:
$$\forall n \in x, \gamma_3''(n)=\nu_2''(\nu_1'(n))\in  \nu_2''(\nu_1'(x))$$
which proves the first rule.

For the second rule, let there be $n, m\in\mathcal{N}_A$ such that

$$\text{parent}(\nu_3''(n))=\text{parent}(\nu_3''(m))$$

Then:

$$\text{parent}(\nu_2''(\nu_1'(n)))=\text{parent}(\nu_2''(\nu_1'(m)))$$

Since $\gamma_2$ is a morphism, we have

$$\text{parent}(\nu_1'(n))=\text{parent}(\nu_1'(m))$$

Since $\gamma_1$ is a morphism, we have

$$\text{parent}(n)=\text{parent}(m)$$

which proves the second rule.

For the third rule, let there be $n\in\mathcal{N}_A$ and $t''\in\mathcal{N}_B''$ such that $\nu_3''(n)\smile t''$.
Thus, $$\nu_2''(\nu_1'(n))\smile t''$$
Because $\gamma_2$ is a morphism, we have
$$\beta_{2,\nu_1'(n)}'(\sigma''(\nu_2''(\nu_1'(n)),t'')=\sigma'(\nu_1'(n),\nu_2'(t'')))$$
In particular we have $$\nu_1'(n)\smile\nu_2'(t'')$$
Because $\gamma_1$ is a morphism, we have
$$\beta_{1,n}(\sigma'(\nu_1'(n),\nu_2'(t''))=\sigma(n,\nu_1(\nu_2'(t'')))$$
Or if we combine:
$$\beta_{1,n}(\beta_{2,\nu_1'(n)}'(\sigma'(\nu_2''(\nu_1'(n)),t'')))=\sigma(n,\nu_1(\nu_2'(t'')))$$
and thus:
$$\beta_{3,n}(\sigma''(\nu_3''(n),t''))=\sigma(n,\nu_3(t''))$$

For the fourth rule, take some $t''\in\mathcal{D}_{\nu_3}$ and take $t=\nu_3(t'')=\nu_1(\nu_2'(t''))$ as well as $t'=\nu_2'(t'')$ which is in $\mathcal{D}_{\nu_1}$.

$\nu_1'$ implies a simplicial map from $\chi(t)\rightarrow\chi'(t')$. $\nu_2''$ implies a simplicial map from $\chi'(t)\rightarrow\chi''(t'')$. Thus $\nu_3''=\nu_2''\circ\nu_1'$ implies a simplicial map from $\chi(t)\rightarrow\chi''(t'')$ by composing the two simplicial maps.

Let there be $n\in\mathcal{N}_A$ and $t''\in\mathcal{D}_{\nu_3}$ such that $\nu_3(t'')\smile n$.
Thus, $$\nu_1(\nu_2'(t''))\smile n$$
Because $\gamma_1$ is a morphism, we have
$$\nu_1'(\sigma(\nu_1(\nu_2'(t'')),n)=\sigma'(\nu_2'(t''),\nu_1'(n)))$$
In particular we have $$\nu_2'(t'')\smile\nu_1'(n)$$
Because $\gamma_2$ is a morphism, we have
$$\nu_2''(\sigma'(\nu_2'(t''),\nu_1'(n))=\sigma''(t'',\nu_2''(\nu_1'(n)))$$
Or if we combine:
$$\nu_2''(\nu_1'(\sigma(\nu_1(\nu_2'(t'')),n))=\nu_2''(\sigma'(\nu_2'(t''),\nu_1'(n))))=\sigma''(t'',\nu_2''(\nu_1'(n))))$$
$$\nu_3''(\sigma(\nu_3(t''),n))=\sigma''(t'',\nu_3''(n)))$$
and thus:
$$\beta_3(\sigma''(\nu_3''(n),t''))=\sigma(n,\nu_3(t''))$$
which proves the fourth rule.

For the fifth rule, let there be $x''\in\mathcal{I''}_A$ and $t''\in\text{support}(\nu_3)$ such that $x''\smile t''$.

It follows that $t''\in\text{support}(\nu_2')$.

Because $\gamma_2$ is a morphism, by applying the fifth rule, there is some $x'\in\mathcal{I'}_A$ such that
$$\nu_2''(x')=x''$$

Furthermore, $t'=\nu_2'(t'')$ is in the support of $\nu_1$ so that applying the fifth rule to morphism $\gamma_1$, there is $x\in\mathcal{I}_A$ such that $$\nu_1'(x)=x'$$

By combining, it follows that $x$ also fulfills the conditions:

$$\nu_3''(x)=\nu_2''(\nu_1'(x))=\nu_2''(x')=x''$$

Thus, $\gamma_3$ is a morphism.

\end{proof}

\begin{lemma}[Associativity of morphism composition]
$\circ$ is associative
\end{lemma}

\begin{proof}
If we take four games $\mathcal{G}$, $\mathcal{G'}$, $\mathcal{G''}$, and $\mathcal{G'''}$ and three morphisms $\gamma_1:\mathcal{G'}\rightarrow\mathcal{G}$ with components $(\nu_1',\beta_1)$, $\gamma_2:\mathcal{G''}\rightarrow\mathcal{G'}$ with components $(\nu_2'',\beta_2')$, and $\gamma_3:\mathcal{G'''}\rightarrow\mathcal{G''}$ with components $(\nu_1''',\beta_1'')$ then $(\gamma_3\circ\gamma_2)\circ\gamma_1$ is by definition:

\begin{itemize}
\item $\nu_3''=(\nu_3''' \circ \nu_2'' )\circ \nu_1'$.
\item $\beta_{3,x}=\beta_{1,x} \circ (\beta_{2,\nu_1'(x)}' \circ \beta_{3,\nu_2''(\nu_1'(x))}'')$
\end{itemize}

which is the same as $\gamma_3\circ(\gamma_2\circ\gamma_1)$:

\begin{itemize}
\item $\nu_3''=\nu_3''' \circ (\nu_2'' \circ \nu_1')$.
\item $\beta_{3,x}=(\beta_{1,x} \circ \beta_{2,\nu_1'(x)}') \circ \beta_{3,\nu_2''(\nu_1'(x))}''$
\end{itemize}

\end{proof}

\begin{remark} A morphism $\gamma:\mathcal{G'}\rightarrow\mathcal{G}$ with components $(\nu',\beta)$ is a section if $\nu'$ is a retraction and $\beta$ is a section. It is a retraction if $\nu'$ is a section and $\beta$ is a retraction.
\end{remark}

\begin{theorem}[Category of games]
The class of all alternating games, endowed with the morphisms as defined above, is a category. We denote it $\textbf{Game}$.
\end{theorem}
\begin{proof} This follows from the left and right unit laws and the associativity of morphism composition.
\end{proof}

\clearpage
\subsection{Causal contextuality scenarios}

We now turn to causal contextuality scenarios. The correspondence table between the two frameworks shown in Table \ref{tab-correspondence} may be useful to the reader to understand how they relate to each other. Note that, in particular, the existence of global section for a set of compatible local (probabilistic or possibilistic) distributions on the event presheaf corresponds to the applicability of the Nash game theory framework, i.e., of the existence of a mixed strategy that is consistent with the observed (probabilistic or possibilistic) distributions.

\begin{sidewaystable}
\noindent
\begin{tabular}{|l|l|}
\hline
\textbf{Spacetime games}&\textbf{Causal contextuality scenarios}\\
\hline
\hline
Alfred&The nature player\\
\hline
Bob&The experimenter\\
\hline
An information set played by Alfred&A measurement\\
\hline
A node played by Alfred&A measurement and its context\\
\hline
An information set played by Bob&A set of events, which is the left-hand-side of an enabling relation\\
\hline
An action by Alfred&A measurement outcome\\
\hline
An action by Bob&A measurement context\\
\hline
An edge from Bob to Alfred&Bob picks a measurement context and carries out a measurement in this context.\\
&The label is the context.\\
\hline
An edge from Alfred to Bob&An outcome $o$ is obtained for a specific measurement $x$, and is part of enabling
\\&conditions $t$ for another measurement. The label is $t(x)=o$\\
\hline
A complete history&A complete history\\
\hline
A history with an even number of defined decisions&A history\\
\hline
An alternating game&A causal contextuality scenario\\
\hline
An alternating game that is in normal form&A flat contextuality scenario\\
\hline
An alternating game whose extensive form has perfect information&A temporal contextuality scenario\\
\hline
The reduced strategic form of the spacetime game&Flattening construction\\
\hline
A pure strategy of the strategic form (under Nash semantics)&A global section of the (pure) strategy sheaf resp. event sheaf\\
\hline
A mixed strategy of the strategic form (under Nash semantics)&The global section of an empirical model (if it exists)\\
\hline
(Dupuysian) causal bridge&Causally secured\\
\hline
The projection of a complete history to Bob&A ``good') cover\\
\hline
\end{tabular}
\caption{The correspondence of terminology between alternating spacetime games and causal contextuality scenarios.}
\label{tab-correspondence}
\end{sidewaystable}

\subsubsection{Definition}

We recall the definition of a causal contextuality scenario with a cover, as given by Abramsky et al. \cite{Abramsky2024}:

\begin{definition}[Causal contextuality scenario]
A causal contextuality scenario with a cover $\Gamma$ is a quadruplet $(X, O, \vdash,\mathcal{C})$ where X is a set of measurement settings, $O=(O_x)_{x\in X}$ is a family of sets of measurement outcomes, indexed by settings, $\vdash$ is an enabling relation between a section on a subset of $X$ (a ``consistent set of events'') and settings, and $\mathcal{C}$ is the set of the facets\footnote{In other words, the cover is an antichain. It is, however, possible to extend the mapping naturally even if the antichain property does not hold (e.g., by allowing unused actions at Bob's nodes).} of a simplicial complex on $X$.
\end{definition}

We use the notation $\textsf{enabled}(t)=\{x\in X|t\vdash x\}$.

\begin{definition}[Unique causal bridges]
A causal contextuality scenario $\Gamma=(X, O, \vdash,\mathcal{C})$ is said to have unique causal bridges if
$$\forall x\in X, \forall t, u, (t\vdash x \wedge u \vdash x \implies t=u)$$
\end{definition}

For causal contextuality scenarios with unique causal bridges, we write $\tau(x)$ for the unique $t$ such that $t \vdash x$. We write $\bar{\tau}(x)$ for the transitive closure of $\tau(x)$ to the set of all past events that must happen to enable measuring $x$. It is obtained by recursively walking up through all enabling relations until one reaches those with an empty left-hand-side, and by merging all sets of events on the left-hand sides of the encountered enabling relations.

In this paper, we consider ``clean'' scenarios in the sense that there are no unused measurement settings, i.e., settings that cannot be enabled under any circumstances under the considered cover. This amounts to saying that $\bar{\tau}(x)$ is always a \emph{consistent} set of events. If a causal contextuality scenario contains an unused setting, then it can be removed without affecting its semantics. This is analogous to our assumption that spacetime games do not have unused information sets.

\begin{definition}[Local cover restriction]
We define the local cover restriction $\mathcal{C}_t$ at an enabling set of events $t$ the set of the facets of the restriction to $\textsf{enabled}(t)$ of the simplicial map associated with $\mathcal{C}$.

\end{definition}

\begin{remark}
We have
$$\forall x\in X, \exists C\in\mathcal{C}_{\tau(x)}, x\in C$$
\end{remark}

\cite{Abramsky2024} mentions that one can define ``good covers'' such that the strategy presheaf fulfills the sheaf property, but a definition thereof is not known to us as we could not find it in publicly available repositories. Since this nice property also characterizes those covers that make a causal contextuality scenario equivalent in structure to an alternating spacetime game, we give a definition of such covers that leads to the categorical equivalence of the two structures. We use the terminology ``causally-secured cover'' which was used in recent presentations by \cite{Abramsky2024}, with the warning that this definition may differ from future publications from the causal contextuality scenario research team.

\begin{definition}[Causally-secured cover]
\label{definition-causally-secured-cover}
A causal contextuality scenario $\Gamma=(X, O, \vdash,\mathcal{C})$ is said to have a causally secured cover if

\begin{itemize}
\item Facet membership propagates to all enabling measurements:
$$\forall C\in\mathcal{C}, \forall C\, \text{support}(\tau(x))\subseteq C$$
\item Two measurements with inconsistent enabling conditions never belong to the same facet:
$$\forall x,y \textsf{ such that } (\bar{\tau}(x)\cup\bar{\tau}(y)) \text{ is inconsistent }, \forall C\in\mathcal{C}, \{x,y\}\not\subseteq C$$
\item $\mathcal{C}$ is the maximum among all covers $\mathcal{C}'$ that have the same local cover restrictions ($\forall t, \mathcal{C}_t'=\mathcal{C}_t$) and that fulfill the other two criteria\footnote{This criterion is akin to some form of non-signaling in the available choices of local covers at any spacetime location.}.
\end{itemize}
\end{definition}

\begin{remark}
By definition, it means that a causally-secured cover is fully specified by all its local cover restrictions.
\end{remark}

\begin{remark}
It follows from the maximality of a causally secured cover that the local cover restrictions can be characterized with
$$\mathcal{C}_t=\{C\cap \textsf{enabled}(t)|C\in\mathcal{C}\}\setminus \{ \emptyset \}$$
that is, it is not necessary to explicitly eliminate non-maximal faces.
\end{remark}

\subsubsection{Category of causal contextuality scenarios}

We now extend the morphism of \cite{Abramsky2019}, which was defined on flat scenarios, to acyclic causal contextuality scenarios with causally secured covers and unique causal bridges (thereafter scenario), i.e., to make it compatible with the enabling relation. We are not aware of a follow-up to \cite{Abramsky2019} that defines the morphisms in the causal case.

\begin{definition}[Scenario morphism]
Consider two scenarios $\Gamma=(X, O, \vdash,\mathcal{C})$ and $\Gamma'=(X', O', \vdash',\mathcal{C}')$. A morphism $\gamma:\Gamma'\rightarrow\Gamma$ is a tuple $(\pi',\alpha=\{\alpha_x\}_{x\in X})$ where:

\begin{itemize}
\item $\pi'$ is a function from $\mathcal{X}$ to $\mathcal{X'}$ that entails a simplicial map from $\mathcal{C}$ to $\mathcal{C'}$.
\item for each $x\in X$, $\alpha_x$ is a function from $O'_{\pi(x)}$ to $O_x$.
\end{itemize}

With the following structural constraints:
\begin{enumerate}
\item $\pi'$ preserves causal bridges:
$$\forall x, y, \tau(\pi'(x))=\tau(\pi'(y))\implies \tau(x)=\tau(y)$$

\item Measurement enablement is not precluded by $\pi'$, i.e., for any $x\in X$, any $y'$ in the support of $\tau(\pi'(x))$ must be in the image of $\pi'$.
$$\forall x\in X, \forall y'\in\textsf{support}(\tau(\pi'(x)), \exists y\in X, \pi'(y)=y'$$

\item For any $x\in X$ and any $y\in X$ such that $\pi'(y)\in\textsf{support}(\tau(\pi'(x))$:
$$\alpha_y\left[\tau(\pi'(x))(\pi'(y)\right]=\tau(x)(y)$$

\end{enumerate}
\end{definition}

\begin{remark}
Rules 2 and 3 imply that, for any $x\in X$, if $\tau(x)=\{(x_i,o_i)\}_i$, then $\tau(\pi'(x))$ can be written in the form $\{(\pi'(x_i),o_i')\}_i$ where:
$$\forall i, \alpha_{x_i}(o_i')=o_i$$

Because all enabling relations must have a consistent set of events as their left-hand-side, the existence of an enabling relation of that form also implies in rule 1:
$$\forall i,j, \pi'(x_i)=\pi'(x_j)\implies o_i'=o_j'$$
\end{remark}

\begin{remark}
\label{remark-scenario-morphism}
Rule 3 also implies that if $t\vdash x$ and $t\vdash y$ then $\tau(\pi'(x))=\tau(\pi'(y))$ because the form of $\tau(\pi'(x))$ depends only on $\tau(x)$ and otherwise not on the particular choice of $x$. Then there is $t'$ such that $t'\vdash \pi'(x)$ and $t'\vdash\pi'(y)$.

This implies that the causal bridge conservation implication is in fact an equivalence:
$$\forall x, y, \tau(\pi'(x))=\tau(\pi'(y))\iff \tau(x)=\tau(y)$$

\end{remark}

There is a scenario morphism that turns any scenario into itself:

\begin{definition}[Identity morphism for a scenario with a causally secured cover]
Consider a scenario $\Gamma=(X, O, \vdash,\mathcal{C})$ . The identity morphism $\text{Id}_\Gamma:\Gamma\rightarrow\Gamma$ is the tuple $(\pi',\alpha)$ where:

\begin{itemize}
\item $\pi'=\text{Id}_{X}$.
\item $\alpha_x=\text{Id}_{O_x}$ for each $x\in X$.
\end{itemize}
\end{definition}
\begin{lemma} The identity morphism for any scenario is a morphism.
\end{lemma}

\begin{proof}
It is straightforward to verify that the identity morphism follows the required structural constraints as they transform into tautologies.
\end{proof}

\begin{lemma}[Left and right unit laws]
Consider two scenarios with no cycles, unique causal bridges, and a causally secured cover $\Gamma=(X, O, \vdash,\mathcal{C})$ and $\Gamma'=(X', O', \vdash',\mathcal{C}')$ as well as a morphism $\gamma:\Gamma'\rightarrow\Gamma$ with components $(\pi',\alpha)$. Then $\gamma=\text{Id}_{\Gamma}\circ\gamma=\gamma\circ\text{Id}_{\Gamma'}$
\end{lemma}
\begin{proof} It follows from applying the left and right unit laws on the individual components.
\end{proof}

And morphisms can be composed like so:

\begin{definition}[Composition of morphisms on scenarios]
Consider three scenarios $\Gamma$, $\Gamma'$, and $\Gamma''$ and two morphisms $\mu_1:\Gamma'\rightarrow\Gamma$ with components $(\pi_1',\alpha_1)$ and $\mu_2:\Gamma''\rightarrow\Gamma'$ with components $(\pi_2'',\alpha_2')$. The morphism $\mu_3=\mu_1 \circ \mu_2:\Gamma''\rightarrow\Gamma$ is defined as the tuple $(\pi_3'',\alpha_3)$ where:

\begin{itemize}
\item $\pi_3''=\pi_2''\circ\pi_1'$.
\item $\alpha_{3,x}=\alpha_{1,x}\circ \alpha_{2,\pi_1'(x)}'$ for every $x\in X$.
\end{itemize}
\end{definition}

\begin{lemma}
The result of morphism decomposition as defined above is a morphism.
\end{lemma}

\begin{proof}
Let $x,y\in X$ such that $$\tau(\pi_3''(x))=\tau(\pi_3''(y))$$

Then we have

$$\tau(\pi_2''(\pi_1'(x)))=\tau(\pi_2''(\pi_1'(y)))$$

Since $\mu_2$ is a morphism:

$$\tau(\pi_1'(x))=\tau(\pi_1'(y))$$

Since $\mu_1$ is a morphism:

$$\tau(x)=\tau(y)$$

which proves the first rule.

Let $x\in X$ and $y''\in\textsf{support}(\tau(\pi_3''(x))$.

It follows that $y''$ is also in the support of $\tau(\pi_2''(\pi_1'(x))$).

Because $\mu_2$ is a morphism, there is some $y'\in X'$ such that $\pi_2''(y')=y''$.

Rule 3 applied $\mu_2$, $\pi_1'(x)$ and $y'$ implies that $y'$ is in the support of $\tau(\pi_1'(x))$.

Since $\mu_1$ is a morphism, there is some $y\in X$ such that $\pi_1'(y)=y'$.

It follows that $\pi_3''(y)=\pi_2''(\pi_1'(y))=\pi_2''(y')=y''$.

This proves the second rule.

Let there be $x\in X$ and $y\in X$ so that $\pi_3''(y)\in\textsf{support}(\tau(\pi_3''(x))$.

Denote $y'=\pi_1'(y)$ and $x'=\pi_1'(x)$.

Then $\pi_2''(y')=\pi_3''(y)\in\textsf{support}(\tau(\pi_3''(x))=\textsf{support}(\tau(\pi_2''(x'))$.

Because $\mu_2$ is a morphism, we have

$$\alpha_{2,y'}(\tau(\pi_2''(x'))(\pi_2''(y'))=\tau(x')(y')=\tau(\pi_1'(x))(\pi_1'(y))$$

Then $\pi_1'(y)$ is in the support of $\tau(\pi_1'(x))$ and because $\mu_1$ is a morphism, we have

$$\alpha_{1,y}(\tau(\pi_1'(x))(\pi_1'(y))=\tau(x)(y)$$

Combining the two, we get:

$$\alpha_{1,y}(\alpha_{2,\pi_1'(y)}(\tau(\pi_2''(\pi_1'(x)))(\pi_2''(\pi_1'(y))))=\tau(x)(y)$$

Which is to say:

$$\alpha_{3,y}(\tau(\pi_3''(x))(\pi_3''(y)))=\tau(x)(y)$$

which proves the third rule.

Thus, $\gamma_3$ is a morphism.
\end{proof}

\begin{lemma}[Associativity of morphism composition]
$\circ$ as defined on scenario morphisms is associative.
\end{lemma}

\begin{proof}
If we take four games $\Gamma$, $\Gamma'$, $\Gamma''$, and $\Gamma'''$ and three morphisms $\gamma_1:\Gamma'\rightarrow\Gamma$ with components $(\pi_1',\alpha_1)$, $\gamma_2:\Gamma''\rightarrow\Gamma'$ with components $(\pi_2'',\alpha_2')$, and $\gamma_3:\Gamma'''\rightarrow\Gamma''$ with components $(\pi_3''',\alpha_3'')$ then $(\gamma_3\circ\gamma_2)\circ\gamma_1$ is by definition:

\begin{itemize}
\item $\pi_4''=(\pi_3''' \circ \pi_2'' )\circ \pi_1'$.
\item $\alpha_{4,x}=\alpha_{1,x}\circ(\alpha_{2,\pi_1'(x)}'\circ \alpha_{3,\pi_2''(\pi_1'(x))}'')$.
\end{itemize}

which is the same as $\gamma_3\circ(\gamma_2\circ\gamma_1)$:

\begin{itemize}
\item $\pi_4''=\pi_3''' \circ (\pi_2'' \circ \pi_1')$.
\item $\alpha_{4,x}=(\alpha_{1,x}\circ\alpha_{2,\pi_1'(x)})'\circ \alpha_{3,\pi_2''(\pi_1'(x))}''$.
\end{itemize}

\end{proof}

\begin{remark} A morphism $\gamma:\Gamma'\rightarrow\Gamma$ with components $(\pi',\alpha)$ is a section if $\pi'$ is a section and $\alpha$ is a retraction. It is a retraction if $\pi'$ is a retraction and $\alpha$ is a section.
\end{remark}

\begin{theorem}[Category of scenarios]
The class of all causal contextuality scenarios with no cycles, unique causal bridges, and a causally-secured cover, endowed with the morphisms as defined above, is a category. We denote it $\textbf{Scenario}$.
\end{theorem}
\begin{proof} This follows from the left and right unit laws and the associativity of morphism composition.
\end{proof}

\clearpage
\subsection{Categorical equivalence}

Our goal is to prove the following theorem:

\begin{theorem}
\label{theorem-equivalence}
The categories $\textbf{Game}$ and $\textbf{Scenario}$ are equivalent:
$$\textbf{Game}\simeq\textbf{Scenario}$$
\end{theorem}

For this, we need to define a functor that is essential surjective and fully faithful.

\subsubsection{Game-Scenario functor}

We now define the functor $F:\textbf{Game}\rightarrow\textbf{Scenario}$.

First, we map the objects. Any alternating spacetime game can be associated with a scenario as follows. As explained in the previous tabular correspondence, the general idea of the overall mapping is that:

\begin{itemize}
\item Alfred's information sets are measurements;
\item Bob's (singleton) information sets are the left-hand sides of the enabling relations of the scenario;
\item Bob's actions (which label edges from Bob to Alfred) are contexts. Alfred's nodes within the same information set are measurements with the same measurement setting but in different contexts and possibly in different causal pasts (i.e., different parent nodes);
\item Alfred's actions (which label edges to Bob) are the measurement outcomes.
\end{itemize}

In order to define the cover of the associated scenario, we first show that each spacetime game (not necessarily alternating) has a natural cover that corresponds to the facets of all combinations of measurements that are performed in one of the complete histories:

\begin{definition}[Natural cover associated with a game]
Given a spacetime game $\mathcal{G}$ with nature player Alfred, the natural cover $\mathcal{C}$ is defined as follows: For each outcome $z\in \mathcal{Z}$, we take the subset $C_z=\mathcal{I}_A\cap\textsf{support}(z)$ of its domain of definition with only and exactly the information sets that are played by Alfred. Then the natural cover is $\mathcal{C}=\textsf{facets}(\{C_z | z\in \mathcal{Z}\}$) where $\textsf{facets}$ eliminates the non-maximal elements, i.e., returns the facets of the simplicial complex obtained via the downward closure of its argument. Note that duplicates are implicitly eliminated because of the definition of a set.
\end{definition}

The natural cover associated with a game thus corresponds to the facets of a simplicial complex whose vertices are the universe information sets.

For example, the natural cover of the game shown in Figure \ref{fig1} (or \ref{fig7}) and Figure \ref{fig2} (or \ref{fig8}) is $\{\{XW\},\{XZ\},\{YW\},\{YZ\}\}$, which is the cycle of rank 4.

The natural cover of the spacetime game shown in Figure \ref{fig3} and Figure \ref{fig4} is $\{\{XY\},\{XZ\},\{YZ\}\}$, which is the cycle of rank 3.

The natural cover of the spacetime game shown in Figure \ref{fig5} and Figure \ref{fig6} is $\{\{X\}, \{YW\}, \{YZ\}\}$ (with $\{Y\}$ eliminated by $\textsf{facets}$).

This natural cover will now be used to define the corresponding measurement scenario for any alternating spacetime game.

We now give a formal definition of the above:

\begin{definition}[Scenario associated with an alternating spacetime game]
\label{ref-onedirection}
Given an alternating spacetime game $\mathcal{G}$ which is an object in the \textbf{Game} category, we define the equivalent scenario $F(\mathcal{G})=(X,O,\vdash, \mathcal{C})$ as follows:
\begin{itemize}
\item $X=\mathcal{I}_A$
\item for each $x\in X$, $O_x=\chi(x)$
\item for each $x\in X$, with $t=\textsf{parent}(x)$, we define that $\hat{t}\vdash x$ where $$\hat{t}=\{(\iota(m),  \sigma(m, t))|m\in\mathcal{N}_A, m\smile t\}$$

\item additionally, for each $x\in X$, if there is a node in $x$ with no parent, we also define $\{\}\vdash x$. \item  $\mathcal{C}$ is the natural cover of $\mathcal{G}$.
\end{itemize}
\end{definition}

For it to be a valid scenario, we must prove that the enabling relations have consistent left-hand sides.

\begin{lemma}
The scenario associated with an alternating spacetime game always has consistent sets of events on the left-hand sides of enabling relations.
\end{lemma}

\begin{proof}
AB1 implies that two nodes in $y$, which is played by Alfred, have the same labels and destination nodes. As a consequence, if we have two nodes $m,n$ such that $\iota(m)=\iota(n)=y, m\smile x, n\smile x$, i.e., $m,n\in y$, it follows that $\sigma(m,t)=\sigma(n,t)$, leading to only one assignment for $t(y)=\sigma(y,t)$.
\end{proof}

\begin{remark}
It follows from the lemma, given some $x\in X$ and its parent $t$, that $\sigma(m, t)=\hat{t}=\sigma(y,t)$ for any $m\in y$.
This means in particular that we can write $\hat{t}$ in this equivalent form:
$$\hat{t}=\{(y,  \sigma(y, t))|y\in\mathcal{I}_A, y\smile t\}$$
See also Figure \ref{fig9} for a visual of this pattern.
\end{remark}

We must also prove that the functor maps games with scenarios that have a causally secured cover.

\begin{lemma}
If $\mathcal{G}\in ob(\textbf{Game})$, then $F(\mathcal{G})$ has a causally secured cover.
\end{lemma}

\begin{proof}
Let us start with the first criterion.
Let $F(\mathcal{G})=(X,O,\vdash, \mathcal{C})$.
Let $C\in\mathcal{C}$.
Since $\mathcal{C}$ is a natural cover of $\mathcal{G}$, it follows that there exists a complete history $z$ such that $C_z=C$.
By definition, it follows that $C\subseteq\{x\in\textsf{support}(z) | \rho(x)=\text{Alfred}\}=\textsf{support}(z)\cap\mathcal{I}_A$.
Let $x\in C$. It follows that $x\in\textsf{support}(z)$ and that $\rho(x)=\text{Alfred}$. 
By BA1, it follows that $x$ has exactly one parent node t and thus there is, by definition of $F(\mathcal{G})$, we have
$$\hat{t}=\{(y,  \sigma(y, t))|y\smile t\}$$
and $\hat{t}\vdash x$.
Furthermore, $t$ must be in support of the same history $z$ by definition of the history.
Let $y$ be in the support of $\hat{t}$. Thus, we have $z(y)=\sigma(y,t)=\hat{t}(y)$. By definition of the history, $y$ must also be in the support of the history. Since $\rho(y)=\text{Alfred}$, it follows that $y\in C_z$. This is true for any $y$, so that $\textsf{support}(\tau(x))\subseteq C_z=C$.
This is true for any C, so that $\mathcal{C}$ is a causally secured cover and $F(\mathcal{G})\in \text{obj}(\textbf{Scenario)}$.

For the second criterion: if two measurements have incompatible enabling relations, then they cannot both appear in the support the same history because this would cause an information set in the past to be associated with two distinct values, which is not possible by definition of a history.

For the third criterion: assume the natural cover $\mathcal{C}$ is not a maximum. Then there is another cover $\mathcal{D}$ that has the same local cover restrictions and that has a facet not contained in any facet of $\mathcal{C}$. Let us order the information sets in this facet as a linearization of the game (that is, in a linear order compatible with the DAG of the game) and consider the first information set $x$ in this facet such that the set of previous information sets is still included in a facet $C$ of $\mathcal{C}$ that is the support of some history $z$. This information set has a parent information set $t$ played by Bob corresponding to some enabling relations. But $x$ is a vertex in the local cover restriction of $\mathcal{C}$ at $t$ (which is the same as the local cover restriction of $\mathcal{D}$ at $t$ by assumption). This implies that $z$ is an incomplete history as it contains the support of $t$ but not $x$. This contradicts the assumption and thus, $\mathcal{C}$ is a maximum among all causally-secured covers with the same local cover restrictions.
\end{proof}

\begin{figure}
\begin{center}
\includegraphics[width=0.5\textwidth]{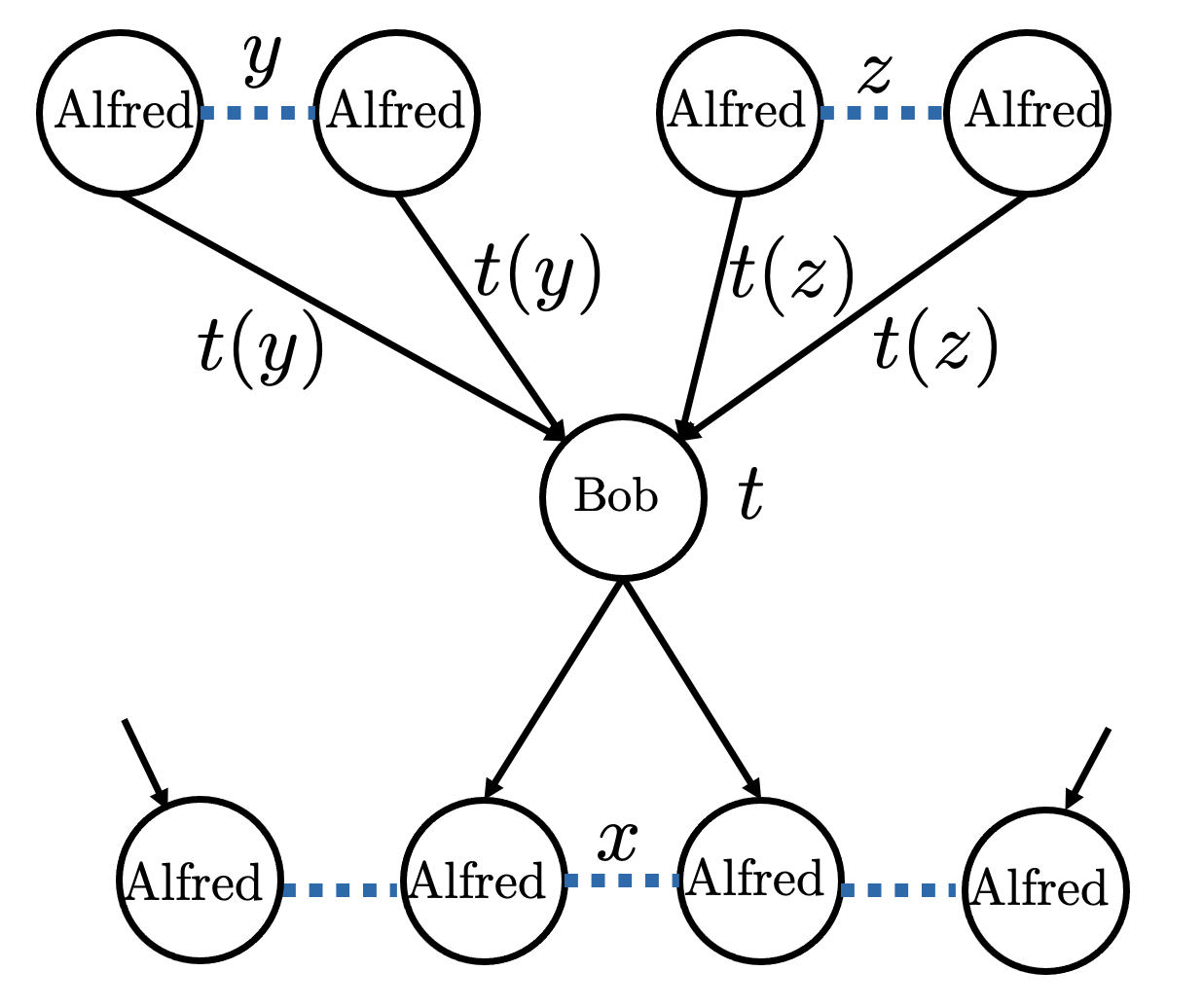}
\end{center}
\caption{The graph pattern implied by the definition of $t\vdash x$ in the mapping from games to scenarios. We show the pattern for two events in $t=\{(y,t(y)),(z,t(z))\}$ two contexts for $y$ and $z$, and four contexts (among which two under causal bridge $t$) for $x$, but this extends to any number of events and any domain size.}
\label{fig9}
\end{figure}

\begin{lemma}
\label{lemma-bijection-tp}
$t\mapsto \hat{t}$ is a bijection between the set of nodes played by Bob and the left-hand sides of the mapped enabling relations. More precisely, for any $x\in X$, the left-hand sides of the enabling relations to $x$ in the scenario are in bijection with the predecessors of $x$ in the alternating spacetime game. When it is clear from the context, we will write $t=\hat{t}$.
\end{lemma}

\begin{proof}
Every $t$ leads to at least one enabling relation with left-hand-side $\hat{t}$ because it has at least one child (property EVEN).
It follows from property AB2 that $t\mapsto\hat{t}$ is injective.
Since the codomain is defined as the range, the mapping is also surjective.
\end{proof}

We now move on to how the functor maps the morphisms.

\begin{definition}
Let $\mathcal{G},\mathcal{G'}\in\textsf{ob}(\textbf{Game})$ (with their components having the usual notations) and let $\gamma:\mathcal{G'}\rightarrow{G}$ be a morphism with components $(\nu',\beta)$. Denoting the components $F(\mathcal{G})$ as $(X,O,\vdash, \mathcal{C})$ and those of $F(\mathcal{G'})$ as $(X',O',\vdash', \mathcal{C}')$,
we define the morphism $F(\gamma):F(\mathcal{G'})\rightarrow F(\mathcal{G})$ as follows.
\begin{itemize}
\item $\pi'=\nu'$ seeing $\nu'$ as acting on information sets. This is mathematically consistent because $X=\mathcal{I}_A$ and $X'=\mathcal{I}_A'$, so that $\pi'$ and $\nu'$ have the same domain and codomain.
\item $\alpha=\beta$. This is also consistent because $\beta$ is a family of maps indexed by $X$  which is the same as $\mathcal{I}_A$, and $\beta_x$ for $x\in\mathcal{I}_A$ is a map from $\chi(\pi'(x))\rightarrow\chi(x)$ which is the same as from $O_{\pi'(x)}'\rightarrow O_x$ by definition of F.
\end{itemize}
\end{definition}

\begin{lemma}
$F(\gamma)$ as defined is a scenario morphism.
\end{lemma}

\begin{proof}
We need to check that each one of the three constraints applies.

Let there be $x,y\in X$ such that $\tau(\pi'(x))=\tau(\pi'(y))$.

By Lemma \ref{lemma-bijection-tp}, $\tau(\pi'(x))=\tau(\pi'(y))$ correspond to one same node $t'\in\mathcal{N}_B$ and by definition of $F$, this node is the parent of both $\pi'(x)$ and of $\pi'(y)$.

Thus $$\textsf{parent}(\pi'(x))=\textsf{parent}(\pi'(y))$$

By definition of F, $\pi'(x)=\nu'(x)$ and $\pi'(y)=\nu'(y)$. Thus, $$\textsf{parent}(\nu'(x))=\textsf{parent}(\nu'(y))$$

Because $\gamma$ is a morphism it follows that $$\textsf{parent}(x)=\textsf{parent}(y)$$.

It follows that $\tau(x)=\tau(y)$ by definition of $F$.

This proves that the first constraint holds.

Let there be $x\in X$ and $y'\in\textsf{support}(\tau(\pi'(x)))$.

It follows that $y'\smile\textsf{parent}(\pi'(x))$ by definition of $F$ and thus $y'\smile\textsf{parent}(\nu'(x))$ also by definition of $F$.

$\textsf{parent}(\nu'(x))\in\mathcal{D}_{\nu}$ by definition of $\mathcal{D}_{\nu}$. Because $\gamma$ is a morphism,

$$\exists y\in\mathcal{I}_A, \nu'(y)=y'$$

from which it follows by definition of $F$:

$$\exists y\in\mathcal{I}_A, \pi'(y)=y'$$

This proves the second constraint.

Let there be $x,y\in X$ such that $\pi'(y)\in\textsf{support}(\tau(\pi'(x)))$.

It follows that $\pi'(y)\smile\textsf{parent}(\pi'(x))$ by definition of $F$ and thus $\nu'(y)\smile\textsf{parent}(\nu'(x))$ also by definition of $F$.

Because $\gamma$ is a morphism,

$$\beta_y(\sigma'(\nu'(y),\textsf{parent}(\nu'(x)))=\sigma(y,\textsf{parent}(x))$$

By definition of $F$ on scenario morphisms:

$$\alpha_y(\sigma'(\pi'(y),\textsf{parent}(\pi'(x)))=\sigma(y,\textsf{parent}(x))$$

By definition of $F$ on games:

$$\alpha_y\left[\tau(\pi'(x))(\pi'(y))\right]=\tau(x)(y)$$

which proves that the third constraint holds.

Thus, $F(\gamma)$ is a scenario morphism.
\end{proof}

\begin{lemma} F preserves the identity morphism.
\end{lemma}

\begin{proof}
$\text{Id}_{\mathcal{G}}$ is defined as:
\begin{itemize}
\item $\nu'=\text{Id}_{\mathcal{N}_A}$ (identified with $\text{Id}_{\mathcal{I}_A}$ with the abuse of notation)
\item $\beta_x=\text{Id}_{\chi(x)}$ for each $x\in\mathcal{N}_A$
\end{itemize}

If we now consider $F(\text{Id}_{\mathcal{G}})$ defined as above and with components $(\pi', \alpha)$ then $\pi'=\nu'=\text{Id}_{\mathcal{I}_A}=\text{Id}_X$ and for each $x\in X=\mathcal{I}_{A}$, $\alpha_x=\beta_x=\text{Id}_{\chi(x)}=\text{Id}_{O_x}$.

Thus, $F(\text{Id}_{\mathcal{G}})=\text{Id}_{F(\mathcal{G})}$
\end{proof}

\begin{lemma} F preserves morphism composition.
\end{lemma}

\begin{proof}
Let us consider three games $\mathcal{G}$, $\mathcal{G'}$, and $\mathcal{G''}$ in $\text{ob}(\textbf{Game})$ and two morphism $\gamma_1:\mathcal{G'}\rightarrow\mathcal{G}$ and $\gamma_2:\mathcal{G''}\rightarrow\mathcal{G'}$. We need to show that $F(\gamma_1 \circ \gamma_2)=F(\gamma_1)\circ F(\gamma_2)$.

Recall that if $\gamma_1$ has components $(\nu_1',\beta_1)$ and $\gamma_2$ has components $(\nu_2'',\beta_2')$, then the morphism $\gamma_3=\gamma_1 \circ \gamma_2:\mathcal{G''}\rightarrow\mathcal{G}$ is defined as the tuple $(\nu_3'',\beta_3)$ where:

\begin{itemize}
\item $\nu_3''=\nu_2'' \circ \nu_1'$.
\item $\beta_{3,x}=\beta_{1,x} \circ \beta_{2,\nu_1'(x)}'$ for every $x\in\mathcal{I}_A$
\end{itemize}

$F(\gamma_1)$ has components $(\pi_1', \alpha_1)$ where:
\begin{itemize}
\item $\pi_1''=\nu_1'$
\item $\alpha_1=\beta_1$
\end{itemize}

$F(\gamma_2)$ has components ($\pi_2'', \alpha_2'$) where:
\begin{itemize}
\item $\pi_2''=\nu_2''$.
\item $\alpha_2'=\beta_2'$.
\end{itemize}

Thus, by definition of the composition of scenario morphisms,  $F(\gamma_1)\circ F(\gamma_2)$ has components ($\pi_3'', \alpha_3$) where:
\begin{itemize}
\item $\pi_3''=\pi_2'' \circ \pi_1'=\nu_2'' \circ \nu_1'=\nu_3''$.
\item for any $x\in X=\mathcal{I}_A$, $\alpha_{3,x}=\alpha_{1,x}\circ \alpha_{2,\pi_1'(x)}=\beta_{1,x}\circ \beta_{2,\nu_1'(x)}=\beta_{3,x}$
\end{itemize}

This matches exactly the definition of $F(\gamma_3)$.

Thus $F(\gamma_3)=F(\gamma_1 \circ \gamma_2)=F(\gamma_1)\circ F(\gamma_2)$ and F preserves morphism composition.

\end{proof}

\begin{theorem} F is a functor.
\end{theorem}
\begin{proof}
F preserves the identity and composition of morphisms. It is thus a functor from $\textbf{Game}$ to $\textbf{Scenario}$.
\end{proof}

\subsubsection{Essential surjectivity}

In order to show that $F$ is essentially surjective, we build a functor $G:\textbf{Scenario}\rightarrow \textbf{Game}$. This functor will allow us to show that any scenario is isomorphic to a scenario in the image of $\textbf{Game}$.

We now give the mapping in the opposite direction, i.e., how to associate a spacetime game with a scenario.

\begin{definition}[Game associated with a scenario]
\label{ref-otherdirection}
Given a scenario $\Gamma=(X,O,\vdash,\mathcal{C})\in\text{ob}(\textbf{Scenario})$, we define the associated spacetime game $\mathcal{G}=G(\Gamma)$ obtained like so:
\begin{itemize}
\item $\mathcal{P}=\{\text{Alfred},\text{Bob}\}$
\item Nodes played by Bob: $\mathcal{N}_B=\{t|\exists x\in X, t\vdash x\}$
\item Nodes played by Alfred: $\mathcal{N}_A=\{(t,x,c)|t\vdash x, c\in \mathcal{C}_t, x\in c\}$
\item $\mathcal{N}=\mathcal{N}_A\cup\mathcal{N}_B$
\item $\forall (t,x,c)\in\mathcal{N}_A, \rho((t,x,c))=\text{Alfred}$
\item $\forall t\in\mathcal{N}_B, \rho(t)=\text{Bob}$
\item Actions by Bob: $\mathcal{A}_B= \bigcup_{t\in\mathcal{N}_B} \mathcal{C}_t$
\item Actions by Alfred: $\mathcal{A}_A=\bigcup_{x\in X} O_x $
\item $\mathcal{A}=\mathcal{A}_A\cup\mathcal{A}_B$
\item $\forall (t,x,c)\in\mathcal{N}_A, \chi(t,x,c)=O_x$
\item $\forall t\in\mathcal{N}_B, \chi(t)=\mathcal{C}_t$
\item $\mathcal{R}_A=\{(t, (t,x,c))|t\vdash x\}$
\item $\mathcal{R}_B=\{((u,x,c),t)|x\in \textsf{support}(t)\}$
\item $\mathcal{R}=\mathcal{R}_A\cup\mathcal{R}_B$
\item $\forall (t,(t,x,c))\in\mathcal{R}_A, \sigma(t, (t,x,c))=c$
\item $\forall ((u,x,c), t)\in\mathcal{R}_B, \sigma((u,x,c),t))=t(x)$
\item Alfred's information sets: $\mathcal{I}_A=\{(t,x,c)|x\in X\}\equiv X$
\item Bob's information sets: $\mathcal{I}_B=\{\{t\}|t\in \mathcal{N}_B\}\equiv\mathcal{N}_B$
\item $\mathcal{I}=\mathcal{I}_A\cup\mathcal{I}_B$
\item $\mathcal{Z}$ is as defined in Definition \ref{definition3} to make the game valid.
\end{itemize}
\end{definition}

\begin{lemma}
The spacetime game associated with a scenario is alternating: $$\forall \Gamma\in\text{ob}(\textbf{Game}), G(\Gamma)\in\text{ob}(\textbf{Scenario})$$
\end{lemma}

\begin{proof}
The game is valid by construction of its set of histories. Thus, we need to show that it is alternating.
\begin{description}
\item{2-PLAYERS} This is true by construction since $\rho$ can only take two values.

\item{BIPARTITE} This holds by construction, as edges always connect a node in $\mathcal{N}_A$ to a node in $\mathcal{N}_B$ or vice versa.

\item{EVEN} This holds by construction: any node $(t,x,c)$ played by Alfred has $t$ as its parent node, so all root notes must be played by Bob. Likewise, given a node $t$ played by Bob, by construction there is at least one enabling relation $t\vdash x$ that it was taken from. Since $\mathcal{C}_t$ is never empty, $t$ has at least one child node.

\item{BOB-S} Bob's information sets are singletons by construction since $\mathcal{I}_B\equiv\mathcal{N}_B$.

\item{BOB-A} No action at Bob's nodes is unused: given a node $t$ played by Bob and an action at this node, i.e., a context $c$ in $\mathcal{C}_t$, $c$ is used on the edge from $t$ to $x$ for any $x\in c$.

\text{BA1} Nodes played by Alfred have exactly one parent by the construction of $\mathcal{R}_A$ because the parent of $(t,x,c)$ is necessarily $t$.

\text{BA2} Let us take a node played by Bob, which must be the left-hand side $t$ of some enabling relation, and which is connected to two different nodes played by Alfred. These must be of the form $(t,x,c_1)$ and $(t,y,c_2)$ with $t\vdash x$ and $t \vdash y$. Let us also assume that the labels are identical. These labels correspond to the contexts in which $x$ and $y$ are measured, meaning that $x$ and $y$ are performed in the same context $c=c_1=c_2$. But then, $x$ and $y$ must be different, otherwise the nodes would be the same $(t,x,c)=(t,y,c)$. Thus, $(t,x,c_1)$ and $(t,y,c_2)$ are in different information sets $x\neq y$.

\text{AB1} requires that all outgoing edges of the nodes $((u,x,c))_{u,c}$ in one of Alfred's information sets $x$ have the same outgoing edges. This is true from the definition $\mathcal{R}_B=\{((u,x,c),t)|x\in \textsf{support}(t)\}$ as well as of $\forall ((u,x,c), t)\in\mathcal{R}_B, \sigma((u,x,c),t))=t(x)$. Indeed, there is no dependency of any edge and label on the context $c$ or on the parent $u$.

\text{AB2} requires that two distinct nodes played by Bob cannot have the same causal bridges. If we take two of these nodes, they correspond to two left-hand sides of enabling relations $t$ and $s$. Let us further assume (reductio ad absurdum) that the nodes $t$ and $s$ have the same predecessors and the corresponding edges have the same labels. This means by construction of $\mathcal{R}_B$ that $t$ and $s$ have the same domains. Furthermore, $t$ and $s$ exactly match on their full domain, also by construction of the labels, which are of the form $t(x)$ and $s(x)$ for each predecessor $x$. Thus, $t=s$ and BA2 is fulfilled.

\end{description}

Thus, the spacetime game associated with a scenario is always alternating.
\end{proof}

\begin{theorem}
The functor $F:\textbf{Game}\rightarrow\textbf{Scenario}$ is essentially surjective.
\end{theorem}

\begin{proof}
Let $\Gamma$ be a scenario in $\text{ob}(\textbf{Scenario})$ and let us consider $\Gamma'=F(G(\Gamma))$.

Let us also denote $\mathcal{G}=G(\Gamma)$.

We now show that $\Gamma'$ and $\Gamma$ are isomorphic.

First, by construction, $X=\mathcal{I}_A\equiv X'$ so we can identify $X$ and $X'$ and we can define $\pi'=\pi=\text{Id}_X$.

By construction, all local cover restrictions of $\mathcal{C}$ and $\mathcal{C}$ are the same and they are causally secured. Thus, they are both maximal and both identical. It follows that $\pi'$ entails a simplicial map from $\mathcal{C}$ to $\mathcal{C}'$ and that $\pi$ entails a simplicial map from $\mathcal{C}'$ to $\mathcal{C}$.

Second, for any $x'\in X'=X$, $O'_{x'}=\chi(x')=O_{x'}$ so we can define $\alpha_x=\alpha_x'=\text{Id}_{O_x}$ for any $x\in X$.

We now consider the scenario morphism $\gamma:\Gamma'\rightarrow\Gamma$ with components $(\pi', \alpha)$, and the morphism $\gamma':\Gamma\rightarrow\Gamma'$ with the same components $(\pi, \alpha')$. Note that $\gamma$ and $\gamma'$ have in fact the same components, but swapped domain and codomain.

If we show that $\gamma$ and $\gamma'$ are indeed scenario morphisms, and that $\gamma\circ\gamma'=\text{Id}_{\Gamma}$ and $\gamma'\circ\gamma=\text{Id}_{\Gamma'}$, then we have shown that $\Gamma$ and $\Gamma'$ are isomorphic.

First, let us remark that by definition of $F$ and $G$, $\Gamma$ and $\Gamma'$ have the same set of measurements $X=X'$, which correspond to $\mathcal{I}_A$ in $\mathcal{G}$.

Let there be $x,y\in X$ such that\footnote{We use $\tau'$ to clarify that this refers to the enabling relations in $\Gamma'$ to avoid any ambiguity, as $\Gamma$ and $\Gamma'$ may have different enabling relations a priori.} $$\tau'(\pi'(x))=\tau'(\pi'(y))$$

Since $\pi'$ is the identity, it follows that $$\tau'(x)=\tau'(y)$$

By definition of $\Gamma'$ as $F(\mathcal{G})$, $\tau'(x)$ corresponds to some node $t$ in $\mathcal{G}$ that is a parent of $x$, its support is given by the parent information sets of $t$ in $\Gamma'$, and its values $\tau'(x)(x_i)$ by the labels $\sigma(x_i, t)$ on the edges from each information set parent $x_i$ of $t$.

By definition of $\mathcal{G}$ as $G(\Gamma)$, the parent information sets of $t$ in $\mathcal{G}$ are all the measurements in the support of $t$ in $\Gamma$, so the support of $\tau'(x)$ is exactly the support of $\tau(x)$. Furthermore, the label between any parent information set $x_i$ of $t$ is, which is $\tau'(x)(x_i)$, is also by definition of $\mathcal{G}$, $\tau(x)(x_i)$. Thus, $\tau'(x)$ and $\tau(x)$ also have the same values and are thus identical. This shows that $\tau'(x)=\tau(x)$: the enabling relations of $\Gamma$ and $\Gamma'$ are the same. With the same argument $\tau'(y)=\tau(y)$ and by transitivity, $\tau(x)=\tau(y)$.

This proves the first rule.

The second rule follows directly from the fact that $\pi'$ is the identity, thus, any measurement is in the image of $\pi'$.

Let there be $x,y\in X$ such that $\pi'(y)\in\textsf{support}(\tau'(\pi'(x)))$ which is to say $y\in\textsf{support}(\tau'(x))$.

Then, $$\alpha_y\left[\tau'(\pi'(x))(\pi'(y)\right]=\alpha_y\left[\tau'(x)(y)\right]=\alpha_y\left[\tau(x)(y)\right]=\tau(x)(y)$$

This fulfills the third constract. Thus, $\gamma$ is a scenario morphism.

The proof that $\gamma'$ is a morphism proceeds similarly.

Thus, $\Gamma$ and $\Gamma'$ are isomorphic and $F$ is essentially surjective.

\end{proof}

\subsubsection{Fullness}

We now turn to the fullness of the functor. This means that it maps morphisms, for any pair of games, surjectively.

\begin{theorem}
The functor $F:\textbf{Game}\rightarrow\textbf{Scenario}$ is full.
\end{theorem}

\begin{proof}
We need to show that F maps morphisms between two given games surjectively.

Let us consider two games $\mathcal{G}\in\textsf{ob}(\textbf{Game})$ and $\mathcal{G'}\in\textsf{ob}(\textbf{Game})$ and their corresponding scenarios $\Gamma=F(\mathcal{G})=(X,O,\vdash,\mathcal{C})$ and $\Gamma'=F(\mathcal{G'})=(X',O',\vdash',\mathcal{C'})$.
Let us consider a morphism $\mu:\Gamma'\rightarrow\Gamma$ with components $(\pi',\alpha)$. We need to show that there exists some $\gamma:\mathcal{G'}\rightarrow\mathcal{G}$ such that $F(\gamma)=\mu$. 

By definition of $F$, $X=\mathcal{I}_A$, $O_x=\chi(x)$ for every $x\in X$, $X'=\mathcal{I}_A'$, $O_{x'}=\chi'(x')$ for every $x'\in X'$. Because of these equalities, we can soundly define the second component of our game morphism as $\beta=\alpha$.

$\pi'$ naturally maps $\mathcal{I}_A$ to $\mathcal{I}_A'$, but we need to extend it to map $\mathcal{N}_A$ to $\mathcal{N}_A'$ in order to define $\nu'$.

Because of Lemma \ref{lemma-bijection-tp}, there is a bijection between the left-hand-sides of the enabling relations in $\Gamma$ and $\mathcal{N}_B$. There is also a bijection between the left-hand-sides of the enabling relations in $\Gamma'$ and $\mathcal{N'}_B$. The definition of $F$ also implies that $\tau$ is identified with $\textsf{parent}$.

Let $n\in\mathcal{N}_A$, and let us define $t=\text{parent}(n)$, $x=\iota(n)$, $x'=\pi'(x)=\iota(\pi'(n))$, $t'=\tau(x')$.

Recall that the structure of the edges and children below $t$ structurally form a simplicial complex whose vertices are $t$'s children information sets, and whose facets are the actions in $\chi(t$). Likewise the structure of the edges and children below $t'$ structurally form a simplicial complex whose vertices are $t'$'s children information sets, and whose facets are the actions in $\chi'(t')$.

If we apply $\pi'$ to the vertices of the simplicial complex at $t$, it can be extended to a simplicial map that in particular maps the facets in $\chi(t)$ to the facets in $\chi'(t')$.

Thus $\sigma(t,n)$ is mapped to some element $\pi'(\sigma(t,n))\in\chi'(t')$ and since $n\in \sigma(t,n)$ we have $x'=\pi'(n)\in\pi'(\sigma(t,n))$ by definition of a simplicial map.

Since $\pi'(\sigma(t,n))\in\chi'(t')$ and because of the structure of the simplicial complex at $t'$, $\pi'(\sigma(t,n))$ is a facet thereof, and there is exactly one node $n'\in x'$ such that $t'\smile n'$ and $\sigma'(t', n')= \pi'(\sigma(t,n))$. We define $\nu'(n)=n'$.

The morphism $\gamma=(\nu',\alpha)$ is indeed a game morphism:
\begin{itemize}
\item The universe nodes are mapped by $\nu'$ in a way consistent with the information sets by construction since it maps them according to $\pi'$.
\item Because $\mu$ is a scenario morphism, $\pi'$ preserves the causal bridge, and thus:

$$\forall x, y\in X, \tau'(\pi'(x))=\tau'(\pi'(y))\implies \tau(x)=\tau(y)$$

Since $\tau$ is identified with $\textsf{parent}$ by functor $F$, this implies:

$$\forall x, y\in X, \textsf{parent}(\nu'(x))=\textsf{parent}(\nu'(y))\implies \textsf{parent}(x)=\textsf{parent}(y)$$

If we now take $n,m\in \mathcal{N}_A$ and we have $\textsf{parent}(\nu'(m))=\textsf{parent}(\nu'(n))$, this implies that $\textsf{parent}(\nu'(\iota(m)))=\textsf{parent}(\nu'(\iota(n)))$ which implies $\textsf{parent}(\iota(m))=\textsf{parent}(\iota(n))$ which implies $\textsf{parent}(m)=\textsf{parent}(n)$.

This also defines $\nu$ as an implied morphism component.

\item Let n be a universe node in $\mathcal{G}$ with $\iota(n)=y$ and $t'\in\mathcal{D}_{\nu}$ an observer node in $\mathcal{G}$ such that $\nu'(n)\smile t'$. Since $t'\in\mathcal{D}_{\nu}$ there exists some $x\in X$ such that $\textsf{parent}(\nu'(x))=t'$. This in turn implies $\tau(\pi'(x))=\hat{t}'$

Since $\nu'(n)\smile t'$ we also have $\pi'(y)\smile t'$ and thus $\pi'(y)\in\textsf{support}(\hat{t'})=\textsf{support}(\tau(\pi'(x)))$. 

Because $\mu$ is a scenario morphism, we have successively:

$$\alpha_y\left[\tau(\pi'(x))(\pi'(y)\right]=\tau(x)(y)$$

$$\beta_y\left[\tau(\pi'(x))(\pi'(y)\right]=\tau(x)(y)$$

$$\beta_y\left[\sigma(\pi'(y),\textsf{parent}(\pi'(x))\right]=\sigma(y,\textsf{parent}(x))$$

$$\beta_y\left[\sigma(\pi'(y),t')\right]=\sigma(y,\nu(t))$$

$$\beta_y\left[\sigma(\pi'(n),t')\right]=\sigma(n,\nu(t))$$

which proves the third rule.

\item Let t' be an observer node in $\mathcal{D}_{\nu}$. Let n be a universe node in $\mathcal{G}$ such that $\nu(t')\smile n$. $\pi'$ restricted to the children of $\nu(t')$ forms a simplicial map from $\chi(\nu(t'))$ into $\chi(t')$ and by definition of $\nu'(n)$:
$$\pi'(\sigma(\nu(t'), n))=\sigma'(t',\nu'(n))$$

Seeing $\nu'$ as a simplicial map, we also have

$$\nu'(\sigma(\nu(t'), n))=\sigma'(t',\nu'(n))$$
which is the required equality.
\item Let $x'$ be a universe information set in $\mathcal{G'}$ (i.e., a measurement in X') and $t'$ an observer node in $\mathcal{G'}$ such that $x'\smile t'$ with $t'$ corresponds to the left hand side of some enabling relation $t' \vdash \nu'(x)=\pi'(x)$ in $\Gamma'$. It implies that $x'\in\textsf{support}(\tau(\pi'(x))$. Since $\mu$ is a morphism, then there exists some $x\in X=\mathcal{I}_A$ so that $\pi'(x)=\nu'(x)=x'$. This proves rule 5.
\end{itemize}
Thus, $\gamma$ is a morphism and $F(\gamma)=\mu$. Thus, F maps morphisms surjectively.
\end{proof}

\subsubsection{Faithfulness}

We finally turn to the faithfulness of the functor. This means that it maps morphisms, for any pair of games, injectively.

\begin{theorem}
The functor $F:\textbf{Game}\rightarrow\textbf{Scenario}$ is faithful.
\end{theorem}

\begin{proof}

We need to show that F maps morphisms between two given games injectively.

Let us consider two games $\mathcal{G}\in\textsf{ob}(\textbf{Game})$ and $\mathcal{G'}\in\textsf{ob}(\textbf{Game})$ and their corresponding scenarios $\Gamma=F(\mathcal{G})=(X,O,\vdash,\mathcal{C})$ and $\Gamma'=F(\mathcal{G'})=(X',O',\vdash',\mathcal{C'})$.

Let us consider two morphisms $\gamma_1:\mathcal{G'}\rightarrow\mathcal{G}$ with components $(\nu_1',\beta_1)$ (and implicitly $\nu_1$) and $\gamma_2:\mathcal{G'}\rightarrow\mathcal{G}$ with components $(\nu_2',\beta_2)$ (and implicitly $\nu_2$) such that $F(\gamma_1)=f(\gamma_2)=\mu$. We now want to prove that $\gamma_1=\gamma_2$.

Let us denote the components of $\mu$ as $(\pi',\alpha)$.

First, by definition of $F$, $\nu_1'$ and $\nu_2'$ when considered as functions on information sets are both equal to $\pi'$ and thus coincide on information sets. In particular:

$$\{\textsf{parent}(\nu_1'(x))|x\in\mathcal{I}_A\}=\{\textsf{parent}(\nu_2'(x))|x\in\mathcal{I}_A\}$$

which implies:

$$\{\textsf{parent}(\nu_1'(n))|n\in\mathcal{N}_A\}=\{\textsf{parent}(\nu_2'(n))|n\in\mathcal{N}_A\}$$

Thus, $\nu_1$ and $\nu_2$ have the same domain of definition.

Let $x\in\mathcal{I}_A$. Then because $\gamma_1$ is a morphism, we have

$$\nu_1(\text{parent}(\nu'_1(x)))=\text{parent}(x)$$

And thus, by definition of F:

$$\nu_1(\text{parent}(\pi'(x)))=\text{parent}(x)$$

The same reasoning for $\gamma_2$, which is also a morphism, leads to:

$$\nu_2(\text{parent}(\nu'_2(x)))=\text{parent}(x)$$

And thus, by definition of F:

$$\nu_2(\text{parent}(\pi'(x)))=\text{parent}(x)$$

Thus 

$$\forall x\in\mathcal{I}_A, \nu_1(\text{parent}(\pi'(x)))=\nu_2(\text{parent}(\pi'(x)))$$

This means that $\nu_1=\nu_2$. We denote $\nu=\nu_1=\nu_2$ in the remainder of this proof.

Let $x\in X=\mathcal{I}_A$. By definition of F applied to $\gamma_1$, we have $\alpha_x(o')=\beta_{1,x}(o')$ for any o' in $\mathcal{O'}_{\pi'(x)}$. By definition of F applied to $\gamma_2$, we have $\alpha_x(o')=\beta_{x,2}(o')$ for any o' in $\mathcal{O}_{\pi'(x)}$. Thus $\beta_{1,x}(o')=\beta_{2,x}(o')$ for any $o'\in\mathcal{O'}_{\pi'(x)}$ that is for any $o'\in \chi'(\pi'(x))$. Thus, $\beta_1$ and $\beta_2$ coincide for all $x$ and $o'$:

$$\beta_1=\beta_2$$

Consider some $t'$ in $\mathcal{N}_B'$. We saw that the labels in $\chi'(t')$ form a simplicial complex whose vertices are the information sets that are the children of $t'$. To each face $a'=\chi'(t')$ and each vertex $x'\in a'$ corresponds exactly one node $n'\in x'$ such that $\sigma'(t',n')=a'$.

If we now consider $t=\nu(t')$ in $\mathcal{N}_B$, we also have that the labels in $\chi(t)$ form a simplicial complex whose vertices are the information sets that are the children of $t$. To each face $a=\chi(t)$ and each vertex $x\in a$ corresponds exactly one node $n\in x$ such that $\sigma(t,n)=a$.

Furthermore, because $\gamma_1$ is a morphism, $\nu_1'$ seen as a map between information sets forms a simplicial map from $\chi(t)$ to $\chi'(t')$. Likewise, because $\gamma_1$ is a morphism, $\nu_2'$ seen as a map between information sets forms a simplicial map from $\chi(t)$ to $\chi'(t')$.

Since $\nu'_1$ and $\nu_2'$ coincide on all vertices (which are the information sets), it follows that $\nu'_1$ and $\nu'_2$ (seen as simplicial maps $\chi(t)$ to $\chi'(t')$) coincide fully, i.e., they map the faces (actions) in the same way. Consequently, they also map the nodes (which are face-vertex pairs) in the same way. Thus, $\nu_1'$ and $\nu_2'$ coincide on $\mathcal{N}_A$.

\end{proof}

We can now conclude with the proof of Theorem \ref{theorem-equivalence}.

\begin{proof}
F is a functor from $\textbf{Game}$ to $\textbf{Scenario}$ that is essentially surjective and fully faithful. It is thus a witness to the equivalence of $\textbf{Game}$ and $\textbf{Scenario}$.
\end{proof}
\clearpage

\section{Discussion}

We now turn to a few examples showing that looking at a scenario via its associated spacetime game is directly useful and simplifies certain constructs, reasonings, and proofs.

\subsection{Good covers}

In the paper introducing causal contextuality scenarios \cite{Abramsky2024}, it is conjectured that certain covers dubbed ``good covers'' (or ``causally secured covers'' in subsequent talks) characterize certain causal contextuality scenarios for which desirable properties apply---such as the sheaf property. We are not aware, at the time of writing, of a subsequent publication giving a formal definition of these good covers, but show further down that if we define causally secured covers as in this paper, they provide such a candidate in which the sheaf property holds for the (pure) strategy presheaf\footnote{But not to empirical models, which correspond to mixed strategies.}.

\subsection{The pure strategy sheaf of a spacetime game}
\label{strategy-sheaf}

The paper introducing causal contextuality scenarios \cite{Abramsky2024} defines strategies of nature and of the experimenter for a given causal contextuality scenario. Follow-up work \cite{Searle2024} also mentions that these strategies can be expressed as global strategies that map measurements to outcomes, a construct also referred to---in our understanding--- as the ``flattening construction'' in an IQSA keynote speech.

If we view a causal contextuality scenario as its corresponding spacetime game, this game has a corresponding strategic form as we explained in 2020 \citep{Fourny2020}: it is the strategic form of the associated game in extensive form with imperfect information, which was constructed 75 years ago by Harold Kuhn \citep{Kuhn1950}\cite{Rubinstein1994}. A pure strategy in this strategic form provides the desired global strategy function.

Worded differently, Kuhn's strategic form can be used to obtain the presheaf of (pure) strategies similar to a presheaf of events whose global sections are the pure strategies of the (non-reduced) strategic form. The reduced strategic form eliminate some redundancies by grouping pure strategies in equivalence classes, however, the reduced pure strategies would then no longer be global sections.

In the flat case, it was shown by Abramsky and Brandenburger \cite{Abramsky2011} that the presheaf of events (which is the flat counterpart to the presheaf of strategies) has the sheaf property because it involves the gluing of partial functions on discrete space that agree on their overlaps. Thanks to the spacetime game framework, it is straightforward to translate this argument to any acyclic causal contextuality scenarios with unique causal bridges and a causally-secured cover.

If we define the pure strategy presheaf as indicated above, then we note that these strategies are partial functions on a discrete space ($\mathcal{I}_A=X$) to actions. Thus, we can generally consider the restrictions of these partial functions to smaller domains, from which it follows as in shown by Abramsky and Brandenburger \cite{Abramsky2011} that the suchly defined presheaf of strategies has the sheaf property.

This was left as an open question by \cite{Abramsky2024} to determine ``good covers'' for which this is the case, and we have thus shown that Definition \ref{definition-causally-secured-cover} successfully addresses this open question at least for acyclic scenarios with unique causal bridges. The simplicity of the construction of the strategy sheaf of Nature demonstrates the expressive power of the spacetime game framework, itself building on 75 years of game theory literature.

Since the categories $\textsf{Game}$ and $\textsf{Scenario}$ are equivalent, one can also instead define the strategy presheaf of a spacetime game by applying the definition given in \cite{Abramsky2024} to the equivalent scenario. We conjecture that this construction is structurally equivalent to the pure strategy presheaf defined above.

\subsection{The mixed strategy presheaf of a spacetime game}

An original contribution of \cite{Abramsky2024} is the generalization of empirical models to scenarios involving causality. Empirical models are obtained through the composition of a distributions functor with the (pure) strategy sheaf. This can be done on different kinds of semirings leading to (signed or unsigned) probabilistic or possibilistic semantics.

Nash equilibria \citep{Nash1951} and other concepts apply naturally to spacetime games \citep{Fourny2020} through their associated extensive form. The strategic form constructed by Kuhn has pure strategies but also mixed strategies, which are probabilistic mixtures of pure strategies, i.e., a probability distribution over all pure strategies of a given player.

Thus, a corollary insight of this paper is that applying the distribution functor to the (pure) strategy sheaf \cite{Abramsky2024} is akin to introducing mixed strategies as distributions over pure strategies as is classically done in game theory. The (pre)sheaf structure then additionally involves restrictions of pure or mixed strategies to smaller supports or, the other way around, asks whether a compatible set of distributions can be glued back into a global section. The existence of a global section corresponds to the case in which a mixed strategy of Nature successfully explains the obtained correlations with various choices of contexts.

Thus, a natural synonym for the distribution presheaf obtained by composing the distribution functor with the sheaf of pure strategies would be \emph{presheaf of mixed strategies}.

A mixed strategy of nature thus corresponds to a deterministic hidden variable model where $\lambda$ is a pure strategy. Free choice, as commonly assumed in the quantum foundations' literature, translates to nature and the observers independently picking their mixed strategies -- as is done when considering Nash equilibria in classical game theory.

Let us take the game in Figure \ref{fig1} (Figure \ref{fig2} for its extensive form with imperfect information) as an example. It corresponds to the 2-2-2 EPR experiment. If we observe that X and Y have disjoint supports (because Alice either picks basis x or y) and we merge them into a single one (call it X) -- and we do the same for Bob (W and Z are merged into a single variable Z) -- then the formula for computing the probability for the complete history resulting from a combination of a mixed strategy of nature (a distribution over the values of $\Lambda$) and the mixed strategies of observers (distributions over the values of $A$ and $B$) happens to \emph{exactly correspond} to the formula used in the Bell inequality literature \citep{Colbeck2017} for local-realist theories with measurement independence:

$$P_{XZ|AB}=\sum_\lambda P_\Lambda(\lambda) P_{X|A\Lambda} P_{Z|B\Lambda}$$

Figure \ref{fig10} shows how this can also be interpreted as a causal network, with all variables depending deterministically on the mixed strategies, leading to a distribution of the complete histories for each combination of mixed strategies of Alice, Bob, and Alfred. Since in our non-Nashian paradigm \citep{Fourny2019}\citep{Baczyk2023} we posit a symmetry between nature and observers (or equivalently, between settings and outcomes), we add explicit lambdas also for the observers and not only for nature.

\begin{figure}
\begin{center}
\includegraphics[width=0.5\textwidth]{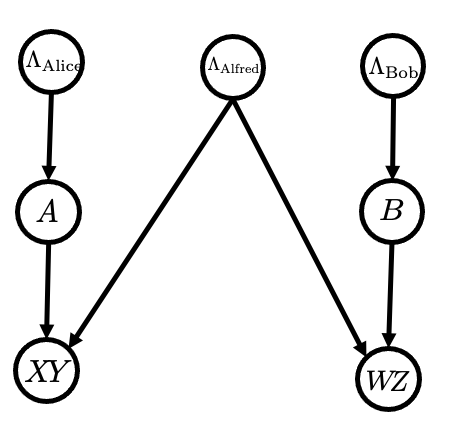}
\end{center}
\caption{A visualization of the Nash mixed strategies as exogenous nodes in a causal network. We merged XY and WZ into single random variables, by observing that the supports are disjoint, to keep the visual simple to understand.}
\label{fig10}
\end{figure}

In the special case of flat games (no adaptive measurements), it was shown by Abramsky and Brandenburger \cite{Abramsky2011} that the existence of a deterministic hidden variable model realizing an empirical model is implied by the existence of a global section for it. With spacetime games, this is intuitive to understand: we observe that a mixed strategy of nature is the same mathematical object as a global section of the event distribution presheaf and we remember that spacetime games have a natural cover given by the facets of the projection of the support of their complete histories to nature.

Then, the following result follows from Harold Kuhn's strategic form:

\begin{theorem}
Given an alternating spacetime game and an empirical model, the existence of a mixed strategy (global section) that is consistent with an empirical model implies the existence of a deterministic hidden variable model for this empirical model.
\end{theorem}

\begin{proof}
If we consider the associated spacetime game, the existence of a global section consistent with the empirical model is equivalent to the existence of a mixed strategy of nature, i.e., a distribution on nature's pure strategies, that is consistent with the empirical model. This mixed strategy provides a deterministic hidden variable model for this empirical model.
\end{proof}

We have just generalized the result of Abramsky and Brandenburger \cite{Abramsky2011} from the flat case to acyclic causal contextuality scenarios with a causally-secured cover and unique causal bridges. The reason for the short proof is not that it is trivial; rather, it builds on previous work including Kuhn's construct. This is consistent with our qualitative argument given by Baczyk and Fourny \cite{Baczyk2023} that Nashian game theory is not compatible with quantum physics. The formulation in terms of presheaves and the sheaf property \cite{Abramsky2011} strengthens this argument.

\subsection{Perfect information}

We observed \cite{Fourny2020}  that imperfect information in the extensive form of the game corresponding to a spacetime game with perfect information can be interpreted in terms of the preservation of locality: information sets are grouped in the way they are because of the finite speed of light, and spacelike separation.

If we do not require perfect information for spacetime games, this introduces another physical justification for imperfect information in the resulting extensive form, namely, noncontextuality because non-singleton information sets in the spacetime game (which lead to the fusion of the corresponding information sets in the extensive form) arise from the different contexts in which a measurement can be carried out, with the result independent of the context. This is also consistent with the fact that non-locality is a special case of contextuality.

\subsection{Temporal contextuality scenarios}

Temporal contextuality scenarios were introduced by Searle et al.  \cite{Searle2024} as a special case of causal contextuality scenario in which the measurements are timelike-separated. If we consider the associated spacetime games, temporal contextuality scenarios correspond to games in extensive form with perfect information. Furthermore, the lookbacks---corresponding to limited memory on past experiments--- can be modeled in the spacetime game framework with imperfect information, i.e., by connecting identical measurements carried out in distinct pasts together. All the above simplified constructs apply: causally-secured covers, pure strategies, the (pure) strategy sheaf, and the link between the existence of a global section of an empirical model (a mixed strategy) and the existence of a deterministic hidden variable model.

\subsection{Information matrices}

Thanks to the equivalence of categories, and since causal contextuality scenarios adapt and repurpose concrete domains and information matrices by Kahn and Plotkin \cite{Kahn1993}, we gain the interesting insight that spacetime games are directly related to this earlier contribution, where cells correspond to nodes of a spacetime game, that is, $\Gamma$ corresponds to $\mathcal{N}$, the set of values $V$ corresponds to the set of actions $\mathcal{A}$, $\mathcal{V}$ corresponds to $\chi$, and the enabling function $\mathcal{E}$ corresponds to the labeled incoming edges of a node ($\mathcal{R}$ and $\sigma$). Further even earlier related work to information matrices includes that of Nielsen, Plotkin and Winskel \cite{Nielsen1981}. Thus, the equivalence of categories connects our work to a dense body of literature, which could lead to further future insights combining the two research directions.

\section{Related work}

Hardy models \citep{Hardy1993}, PR boxes \citep{Popescu1994}, and GHZ models \citep{Greenberger1989} are representable with spacetime games. The framework also maps to contextuality by default \citep{Dzhafarov2016}, which corresponds to having perfect information. Forcing random variables corresponding to measurement outcomes to have the same value formally corresponds to grouping the corresponding nodes into the same information set.

Sheaves of sections as constructed by Gogioso and Pinzani \citep{Gogioso2021} correspond to a special case of spacetime games with perfect information in which measurements are performed independently of previous measurement outcomes, without merging the information sets (see also Section 6.b. in \citep{Abramsky2024}), and with unique causal bridges. The fact that the sheaf property obtains is consistent with our results, which generalize the case covered by Gogioso and Pinzani \citep{Gogioso2021} to all alternating spacetime games, and the causal contextuality scenarios with no cycles, unique causal bridges, and a causally-secured cover.

A definition of free choice is given by Renner and Colbeck \cite{Renner2011} and Sengupta \cite{Sengupta2025}, as a decision (called spacetime variable) being independent of anything it could not have caused. A deterministic resolution algorithm for all spacetime games with perfect or imperfect information, and which circumvents impossibility theorems (such as the sheaf property of the strategy distribution presheaf not generally holding) by weakening free choice, is given in \citep{Fourny2019}, with the special case of games in normal form \citep{Fourny2017} and of games in extensive form with perfect information \citep{Fourny2014}. A high-level discussion of the non-Nashian paradigm in the context of quantum physics interpreted as a game between the observer and nature is also available \citep{Fourny2019q}.

The seeding philosophy of the non-Nashian approach as well as the notion of causal bridge is attributable to Dupuy \cite{Dupuy1992} (counterfactual decision theory) and the game theoretic seed and initial conjectures were made by \cite{Dupuy2000}. The work of Petri \cite{Petri1996} on causally secured events around the same time is independently related to Dupuy's work, and Petri nets, dating back to as early as 1939 in their first form, can also be seen as an early precursor to the concept of causal bridge in a special-relativistic sense.

There is also subsequent work by Kadnikov and Wichardt \cite{Kadnikov2022}\citep{Kadnikov2020} that provides a strategic decision-making model generalizing the normal form and extensive form with graph structures related to spacetime games. This work defines the outcomes of the model based on sets of leaves of the graph structure and nicely connects it to the set-theoretic form by Bonanno \cite{Bonanno1992}\citep{Bonanno1991}. These sets of leaves can be put in correspondence with the complete histories, and the outcomes of the game in extensive form with imperfect information, giving additional and complementary insights.

\cite{Hance2022} give an excellent argument of why impossibility theorems do not disprove the existence of local theories. Hance and Hossenfelder \cite{Hance2022s} also are of the opinion that free choice can be weakened. Wharton and Argaman \cite{Wharton2020} refer to this possibility as (non-)future-input-dependent parameters, and it is also known as measurement independence (see for example the work of Storz et al. \citep{Storz2023} for the explicit mention of this assumption in Bell experiments) or setting independence \citep{Mueller2023}.

\section{Conclusion and future work}

We would like to conclude by sharing our scientific belief that \emph{only acyclic causal contextuality scenarios with unique causal bridges and causally-secured covers correspond to a physical reality}, and that more general scenarios (with cycles, or non-unique causal bridges, or more general covers) are, above all, a useful mathematical tool for reductio ad absurdum proofs that nature is contextual.

Whether scenarios with cycles can exist is the subject of a broad debate in the quantum foundations community known as indefinite causal order. For example, see an argument against the existence of such scenarios in \cite{Vilasini2024}.

Non-unique causal bridges physically mean that a measurement carried out under specific enabling conditions is physically the same measurement as a measurement carried out under different enabling conditions. What this means is closely related to the difference between the classical, Nashian view of free choice as opposed to the non-Nashian approach. Indeed, the core difference between the two approaches is to what extent the past light cone of a decision deterministically determines this decision: the link is more loose in the Nashian view, while it is tight in the non-Nashian view. In the non-Nashian view, which is the direction of our research programme, the tight link between a decision (by an observer or the universe) and the past light cone is consistent with unique causal bridges: a measurement carried out in specific enabling conditions is a not the same as a measurement (counterfactually) carried out in different conditions.

Finally, causally-secured cover pertain to locality as can be seen from the criteria used in our definition: saying that the support of the set of events on the left-hand side of an enabling condition must be included in every context in which the right-hand side appears is consistent with the semantics of enabling conditions. The maximality of a causally-secured cover can be interpreted in terms of a form of non-signaling, in the sense that choices made at some location cannot influence the available contexts at a remote location, and the local cover restrictions are also a feature of locality.

This discussion is aligned with work by Dzhafarov \cite{Dzhafarov2016} and on keeping locality as a fundamental principle of nature, provided we weaken measurement independence and assume non-Nashian free choice instead \cite{Baczyk2023}.

Future work involves a generalization of the categorical equivalence between games and causal contextuality scenarios to a more limited set of assumptions. It also involves further exploration on the consequences of the 

\section{Acknowledgements}

I thank my students at ETH who contributed to the game framework in various projects over the years: Felipe Sulser, Emilien Pilloud, Ramon Gomm, Luc Stoffer, Martin Kucera, Michael Baczyk, Florian Pauschitz. I thank Ben Moseley, Renato Renner, Kuntal Sengupta, Ognyan Oreshkov, Jacopo Surace, Beatrix Hiesmayr, Marek Szopa, Kim Vallée, Stefan Weigert for stimulating discussions and my co-author Michael Baczyk for his valuable feedback on the mapping between games and scenarios.

\bibliography{mapping}

\section*{Declarations}

\subsection*{Funding}
We thank Gustavo Alonso for financing the trip to the IQSA 2024 conference in Brussels.
\subsection*{Competing interests}
The authors have no relevant financial or non-financial interests to disclose.
\subsection*{Author contribution}
The author wrote, read and approved the final manuscript.

\end{document}